\documentclass[11pt,letterpaper]{article}
\usepackage[utf8]{inputenc}
\usepackage{times}
\usepackage[ruled,vlined,commentsnumbered,titlenotnumbered]{algorithm2e}
\usepackage{url}
\usepackage{fullpage}
\usepackage{pslatex} 
\usepackage{mathdots} 
\usepackage{amsmath}
\usepackage{algorithmic}
\usepackage{amssymb}
\usepackage{amsthm}
\usepackage{color}
\usepackage{array}
\usepackage{xy}
\usepackage{setspace}
\usepackage{multicol}
\usepackage{hyperref}
\usepackage{xspace}
\usepackage{cite}
\usepackage[margin=1in,letterpaper]{geometry}
\newcommand{\kicktree}{kick tree\xspace}
\newcommand{\kicktrees}{kick trees\xspace}

\newcommand{\cubby}{cubby\xspace}
\newcommand{\cubbys}{cubbies\xspace}
\newcommand{\facility}{facility\xspace}
\newcommand{\facilitys}{facilities\xspace}

\DeclareMathOperator*{\argmin}{argmin}

\usepackage{accents}

\newcommand{\defn}[1]{\emph{\textbf{{#1}}}}

\usepackage{todonotes}[disable]

\newcommand{\N}{\mathbb{N}}
\newcommand{\poly}{\operatorname{poly}}
\newcommand{\polylog}{\operatorname{polylog}}
\newcommand{\E}{\mathbb{E}}

\newcommand{\Tower}{\operatorname{tow}}

\renewcommand{\paragraph}[1]{\vspace{.5 cm} \noindent \textbf{#1} }

\usepackage{thmtools}
\usepackage{thm-restate}

\newtheoremstyle{slanted}
{3pt}
{3pt}
{\slshape}
{}
{\bfseries}
{.}
{.5em}
{}
\theoremstyle{slanted}
\newtheorem{theorem}{Theorem}
\newtheorem{lemma}[theorem]{Lemma}

\newtheorem{corollary}[theorem]{Corollary}

\begin{document}

\date{}

\title{On the Optimal Time/Space Tradeoff for Hash Tables}

\author{Michael A. Bender\\Stony Brook University \and Mart\'{\i}n Farach-Colton\\Rutgers University \and John Kuszmaul\\Yale University \and  \and William Kuszmaul\\MIT \and Mingmou Liu \\ NTU}

\maketitle

\thispagestyle{empty}
\begin{abstract}
For nearly six decades, the central open question in the study of hash tables has been to determine the optimal achievable tradeoff curve between time and space. State-of-the-art hash tables offer the following guarantee: If keys/values are $\Theta(\log n)$ bits each, then it is possible to achieve constant-time insertions/deletions/queries while wasting only $O(\log \log n)$ bits of space per key when compared to the information-theoretic optimum. Even prior to this bound being achieved, the target of $O(\log \log n)$ wasted bits per key was known to be a natural end goal, and was proven to be optimal for a number of closely related problems (e.g., stable hashing, dynamic retrieval, and dynamically-resized filters).

This paper shows that $O(\log \log n)$ wasted bits per key is not the end of the line for hashing. In fact, for any $k \in [\log^* n]$, it is possible to achieve $O(k)$-time insertions/deletions, $O(1)$-time queries, and 
$$O(\log^{(k)} n) = O\left(\underbrace{\log \log \cdots \log}_{k} n\right)$$
wasted bits per key (all with high probability in $n$). This means that, each time we increase insertion/deletion time by an \emph{additive constant}, we reduce the wasted bits per key \emph{exponentially}. We further show that this tradeoff curve is the best achievable by any of a large class of hash tables, including any hash table designed using the current framework for making constant-time hash tables succinct.

Our results hold not just for fixed-capacity hash tables, but also for hash tables that are dynamically resized (this is a fundamental departure from what is possible for filters); and for hash tables that store very large keys/values, each of which can be up to $n^{o(1)}$ bits (this breaks with the conventional wisdom that larger keys/values should lead to more wasted bits per key). For very small keys/values, we are able to tighten our bounds to $o(1)$ wasted bits per key, even when $k = O(1)$. Building on this, we obtain a constant-time dynamic filter that uses $n \left\lceil \log \epsilon^{-1} \right\rceil + n \log e + o(n)$ bits of space for a wide choice of false-positive rates $\epsilon$, resolving a long-standing open problem for the design of dynamic filters. 

\end{abstract}
\vfill
\newpage 
\pagenumbering{arabic}

\section{Introduction}\label{sec:intro}

A \defn{hash table} \cite{KnuthVol3} (sometimes called a \defn{dictionary}) is a data structure that stores a set of keys from some key-universe $U$ and that supports three operations on that set: insertions, deletions, and queries. Some hash tables are also capable of storing a \defn{value} $v \in V$ associated with each key. In this case, a query on a key $k$ returns both whether key $k$ is present and what the associated value $v$ is, if $k$ is present.

Since hash tables were introduced in 1953, there has been a vast literature on the question of how to design space- and time-efficient hash tables \cite{knuth1963notes, KnuthVol3, Fredman82FKS, dietzfelbinger1988dynamic, dietzfelbinger1990new, pagh2001low, raman2003succinct, demaine2006dictionariis, patrascu2008succincter, arbitman2010backyard, yu2020nearly, liu2020succinct, bender2021all, PaghRo04, FotakisPaSa03, dietzfelbinger1990new, arbitman2009amortized}. Whereas early hash tables \cite{KnuthVol3, knuth1963notes} required $\omega(1)$ time per operation in order to support a load factor of $1 - o(1)$, modern hash tables \cite{arbitman2010backyard, liu2020succinct, bender2021all} offer a much stronger guarantee. Not only are these hash tables constant time (with high probability), and not only do they support a load factor of $1 - o(1)$, but they have even converged towards the \emph{information-theoretically optimal} number of bits of space, given by 
$$\mathcal{B}(U, V, n) = \log \binom{|U|}{n} + n \log |V|.$$ 

A hash table that uses $\mathcal{B}(U, V, n) + rn$ bits of space is said to incur $r$ \defn{wasted bits per key}. When $\log |U| + \log |V| = c \log n$ for some constant $c > 1$, the state of the art for $r$ is $O(\log \log n)$, which was first achieved in 2003 by Raman and Rao \cite{raman2003succinct} with constant expected-time operations, and which after a long line of work \cite{demaine2006dictionariis, arbitman2010backyard, liu2020succinct, bender2021all} has now also been achieved \cite{bender2021all} with (high-probability) worst-case constant-time operations.  

Besides having remained the state of the art for nearly two decades, there are several more fundamental reasons to believe that $r = O(\log \log n)$ might be optimal. It is known that $\Theta(\log \log n)$ wasted bits per key is optimal for the closely related problems of dynamic value retrieval\footnote{A dynamic data-retrieval data structure is a hash table with the added restriction that queries must be for keys/value pairs that are present. If keys are from a universe of size $\poly(n)$ and $v$ is the size of each value in bits, then static value-retrieval requires $nv + o(n)$ bits \cite{dietzfelbinger2008succinct}, and dynamic value-retrieval requires $nv + \Theta(n \log \log n)$ bits \cite{demaine2006dictionariis, mortensen2005dynamic}.} \cite{mortensen2005dynamic, demaine2006dictionariis,dietzfelbinger2008succinct} and fully-dynamic approximate set membership\footnote{Fully-dynamic approximate set membership data structures, also known as dynamically-resizable filters, are analogous to dynamically-resizable hash tables, but with some $\epsilon$ probability of queries returning false-positives. Whereas an optimal static filter requires $n\log \epsilon^{-1}$ bits \cite{carter1978exact}, an optimal resizable filter requires $n \log \epsilon^{-1} + \Omega(n \log \log n)$ space \cite{pagh2013approximate}, which is known to be optimal for $\epsilon^{-1} \le \polylog n$ \cite{pagh2013approximate, liu2020succinct}.} \cite{pagh2013approximate, liu2020succinct,carter1978exact}. And it is known that \emph{stable} hash tables \cite{demaine2006dictionariis, bender2021all} (i.e., hash tables in which each key/value pair is assigned a fixed and unchanging position upon arrival) have an optimal value of $r = \Theta(\log \log n)$. 

Nonetheless, it is \emph{not known} whether $r = O(\log \log n)$ wasted bit per key is optimal for dynamic constant-time hash tables. More generally, it is an open question what the optimal tradeoff is between time and space (e.g., can slightly super-constant-time operations yield major space savings?). Tight answers to these questions would close off one of the longest-standing directions of research in the field of data structures.

\paragraph{An astonishing tradeoff curve between time and space.} In this paper, we present a data structure that achieves a much stronger time/space tradeoff than existing hash tables, and we prove a matching lower bound establishing that our hash table is optimal across a large family of data structures that includes all existing fast succinct hash tables.

We show that it is possible to achieve significantly fewer than $O(\log \log n)$ wasted bits per key. In fact, for any parameter $k \in [ \log^* n]$, we construct a hash table that supports constant-time queries, that supports $O(k)$-time insertions/deletions, and that incurs 
$$O(\log^{(k)} n) = O\left(\underbrace{\log \log \cdots \log}_{k} n\right)$$
wasted bits per key, where the guarantees on time and space are worst-case with high probability in $n$. Our result holds not just for fixed-capacity hash tables, but also for dynamically-resizable hash tables, as well as for hash tables storing very large keys/values (up to $n^{o(1)}$ bits each). 

Our result implies a remarkably steep tradeoff: each time that we increase insertion/deletion time by an \emph{additive constant}, we are able to \emph{exponentially reduce} the number of wasted bits per key. In particular, we obtain a hash table that supports $O(1)$-time insertions/deletions/queries with $O(\log^{(c)} n)$ wasted bits per key, for any positive constant $c$ of our choice; and we obtain a hash table that supports $O(\log^* n)$-time insertions/deletions with $O(1)$ wasted bits per key and constant-time queries. 

As we mentioned above, we prove that the tradeoff curve on which these hash tables sit is tight for a large class of data structures, including all known dynamic succinct hash tables \cite{liu2020succinct, bercea2020dynamic, bender2018bloom, arbitman2010backyard, bender2021all, raman2003succinct}, and more generally, any hash table that makes implicit (or explicit) use of ``augmented open addressing'' (discussed in more detail in Section~\ref{sec:augmented-open-addressing}).

Finally, in the special case where keys/values are small, meaning that each key-value pair consists of $\log n + o(\log n / \log^{(k)} n)$ bits (but keys are still from a universe of size $\omega(n)$), we are able to further tighten our bounds to obtain $o(1)$ wasted bits per key. Building on this, we obtain a dynamic constant-time approximate set-membership data structure (i.e., a filter) that achieves space
$$ n \log \epsilon^{-1} + n\log e + o(n)$$ 
bits, for a wide choice of false-positive rates $\epsilon$, resolving a long-standing open problem as to whether $O(1)$ wasted bits per key is achievable by dynamic filters. In fact, not only is $\log e + o(1)$ constant, but it is the provably optimal number of wasted bits per key for any filter that is constructed by storing fingerprints in a hash table \cite{carter1978exact, pagh2005optimal, bender2018bloom, bercea2020dynamic, liu2020succinct}.

\section{Overview of Results and Techniques}\label{sec:overview}

This section presents an overview of the main results and techniques in the paper.

\subsection{The relationship between hash tables and balls-to-slots schemes}
\label{sec:augmented-open-addressing}

An implicit theme in the design and analysis of hash tables is that the problem of constructing a space-efficient hash table is closely related to the problem of placing balls into slots of an array. We now formalize this relationship by defining the class of \defn{augmented open-addressed hash tables} (which includes all known succinct constant-time hash tables \cite{liu2020succinct, bercea2020dynamic, bender2018bloom, arbitman2010backyard, bender2021all, raman2003succinct}), and by formally defining the balls-to-slots problem that any augmented open-addressed hash table must solve (we will call this problem the \defn{probe-complexity problem}). Later, in Section \ref{sec:probe}, we will give tight upper and lower bounds for the probe-complexity problem, which in subsequent sections will allow for us to construct optimal augmented open-addressed hash tables.

\paragraph{Augmented open addressing.} 
Augmented open-addressed hash tables are hash tables that abide by the following basic framework: elements are stored in an array of some size $m = (1 + \epsilon)n$, and each element $x \in U$ is assigned a probe sequence $h_1(x), h_2(x), h_3(x), \ldots\in [m] $ of array slots where it can be stored; queries are then implemented using a secondary query-routing data structure which, for each key $x$, stores the index $i$ of the position $h_i(x)$ containing the key.\footnote{Additionally, a technique known as ``quotienting'' is used to shave $\log n$ bits off of each key---this is what bridges the gap between using $\mathcal{B}(U, V, n) + nr$ space instead of $n \log |U| + n \log |V| + nr$ space.} If a hash table is to be succinct, it must simultaneously achieve $\epsilon = o(1)$ (we call the quantity $1 - \epsilon$ the \defn{load factor}), while also ensuring that the quantities stored by the query-routing data structure don't take up too many bits (i.e., keys are in positions $h_i(x)$ for relatively small values of $i$).

The use of a probe sequence to determine where a key can reside is analogous to classical open addressing \cite{KnuthVol3}. An important difference is that the query-routing data structure allows for queries to be performed in constant time, without needing to scan through the positions $h_1(x), h_2(x), \ldots$. 

Of course, there is flexibility in terms of what granularity augmented open addressing is used at. For example, in order to support dynamic resizing \cite{raman2003succinct, liu2020succinct, bender2021all}, a hash table might use augmented open addressing on bins of size $\polylog n$, and then use a different set of techniques to determine which bin each key should go into. Nonetheless, all known succinct constant-time hash tables \cite{liu2020succinct, bercea2020dynamic, bender2018bloom, arbitman2010backyard, bender2021all, raman2003succinct} rely on some form of augmented open-addressing as the highest-granularity abstraction layer in which elements are stored. 

\paragraph{The probe complexity problem.} We can formalize the balls-to-slots problem that any augmented open-addressed hash table must solve as follows. Each key $x$ is thought of as a ``ball'' that is associated with some probe sequence $h(x) =  \langle  h_1 (x), h_2 (x),\ldots, h_m(x) \rangle$ where without loss of generality the sequence is a permutation of $\langle 1, 2, \ldots, m\rangle$. A balls-to-slots scheme must support an (online) sequence of ball insertions/deletions so that, at any given moment, the up to $n$ balls that are present are each assigned distinct positions in an array of $m = (1 + \epsilon)$ slots. The balls-to-slots scheme is measured by two objectives: the average \defn{probe complexity} of the balls, which for a ball $x$ in slot $h_i(x)$ is given by $(1 + \log i)$; and the \defn{switching cost} of the balls-to-slots scheme, which is the number of balls that the scheme rearranges on each insertion/deletion (including the ball being inserted/deleted). 

The probe complexity of each ball $x$ can be viewed as the minimum number of bits (asymptotically) that must be stored in the query router in order for the position of the ball to be recovered by queries. The switching cost, on the other hand, can be viewed as (a lower bound on) the amount of time that it takes to implement a ball insertion/deletion. Thus lower bounds on the relationship between probe complexity and switching cost in the probe-complexity problem directly translate to lower bounds on the relationship between space and insertion-time in augmented open-addressed hash tables.

Intuitively, there are three challenges to designing an augmented open-addressed hash table: one must construct a balls-to-slots scheme with low probe complexity, low switching cost, and high load factor; one must efficiently implement that balls-to-slots scheme so that insertions/deletions can (ideally) be performed in time proportional to the switching cost; and one must implement a query-routing data structure that maps each key $x$ to the slot $h_i(x)$ where it resides (ideally, this should use space proportional to the probe complexity of $x$). Thus the problem of determining the optimal tradeoff between average probe complexity and switching cost is central to the problem of designing an optimal augmented open-addressed hash table.

\paragraph{The two approaches to designing balls-to-slots schemes.}
One way to design a balls-to-slots scheme is to base it on a traditional open-addressed hash table such a linear probing \cite{KnuthVol3}, double hashing \cite{KnuthVol3}, or Cuckoo hashing \cite{PaghRo04}. Cuckoo hashing \cite{PaghRo04} (and its variants \cite{FotakisPaSa03, dietzfelbinger2005balanced, arbitman2009amortized}) is especially appealing because it bounds probe complexity deterministically. Standard Cuckoo hashing achieves a probe complexity of $O(1)$, but is only able to support a load factor of $1 - \epsilon < 1/2$. Generalizations of Cuckoo hashing (i.e., $d$-ary Cuckoo hashing \cite{FotakisPaSa03} and Cuckoo hashing with $d$-slot bins \cite{dietzfelbinger2005balanced}) are able to support higher load factors $1 - \epsilon$, but at the cost of incurring a super-constant switching cost of at least $\Omega(\log \epsilon^{-1})$. Thus Cuckoo hashing cannot be used on its own to obtain a succinct constant-time hash table (although it has been used in past work as an essential building block \cite{arbitman2010backyard}). 

To achieve small probe complexity and switching cost, while also supporting $\epsilon = o(1)$, past work has used balls-to-slots schemes that are based on standard balls-to-bins techniques. One simple approach is to set $m = (1 + 1 / \log n)n$; to partition the array of size $m$ into bins of size $\ell = \polylog n$; and finally to hash each key to a random bin $g(x) \in \{0, 1, 2, \ldots, m / \ell - 1\}$ and set $h_i(x) = g(x) \cdot \ell + i$ for all $i$. With high probability in $n$, every key will find a free position in the bin that it hashes to, meaning that each key is assigned to one of its first $\ell = \polylog n$ choices. This scheme achieves load factor $1 - 1 / \log n$, worst-case probe complexity $O(\log \log n)$, and worst-case switching cost $1$ (with high probability in $n$). 

It's natural to hope that an even better probe complexity could be achieved by making use of more sophisticated balls-to-bins schemes (e.g., power of two choices \cite{mitzenmacher2001power}). This turns out not to be possible, as one can prove a lower bound of $\Omega(\log \log n)$ average probe complexity for \emph{any} balls-to-slots scheme with switching cost $1$ and load factor $1 - 1 / \log n$; in fact, this is a special case of a more general lower bound \cite{demaine2006dictionariis, mortensen2005dynamic} which says that \defn{stable hash tables} (i.e., hash tables in which elements are assigned permanent positions when they are inserted) must incur $\Omega(\log \log n)$ wasted bits per key.

The central bottleneck to designing augmented open-addressed hash tables that make use of this simple balls-to-slots scheme has been the issue of achieving constant-time operations while preserving space efficiency \cite{raman2003succinct, arbitman2010backyard, liu2020succinct, bercea2020dynamic, bender2018bloom, bender2021all}. Raman and Rao \cite{raman2003succinct} gave an elegant solution with constant expected time in which the query-routing data structure is itself a collection of small hash tables that store fingerprints of keys. The bottleneck to achieving the same guarantee with worst-case time bounds has been, until recently, the difficulty of constructing efficient query-routing data structures for bins of $\polylog n$ elements---this led researchers to develop more sophisticated balls-to-slots schemes that make use of smaller bins \cite{arbitman2010backyard, liu2020succinct, bercea2020dynamic, bender2018bloom}, which allowed for them to overcome the query-routing bottleneck, but resulted in a worse space utilization. Recently, \cite{bender2021all} resolved this issue by showing how to perform query-routing on bins of size $\polylog n$ while incurring only $O(\log \log n)$ extra wasted bits per key. 

In summary, it is known how to use the balls-to-bins framework for the probe-complexity problem in order to achieve $O(\log \log n)$ wasted bits per key, and this results in an optimal stable hash table. It has not been known whether the ability to move keys around during insertions opens the door to even higher space efficiency.

\paragraph{An optimal solution to the probe-complexity problem.}
An essential technical insight in our paper is that one can achieve an extremely small average probe complexity by moving around just a few balls on each insertion. We present a balls-to-slots scheme, called the \defn{$k$-\kicktree}, that achieves average probe complexity $\log^{(k)} n$ while achieving a worst-case switching cost of $O(k)$. Moreover, this result holds even when the number $n$ of balls equals the number $m$ of slots, so the load factor is $1$. 

We prove that this tradeoff between switching cost and probe complexity is asymptotically optimal (as long as the load factor $1 - \epsilon$ is at least, say, $1 - 1 / \log \log n$). In particular, if a balls-to-slots scheme achieves average probe complexity $O(\log^{(k)} n)$, it must move an average of $\Omega(k)$ items per insertion/deletion. This lower bound is established via an intricate potential-function argument that we consider to be one of the main technical contributions of the paper. 

Interpreting our lower bound as a statement about augmented open-addressed hash tables, we can conclude that if there exists a hash table with a better time/space tradeoff curve than the hash tables in this paper, it would have to fundamentally avoid the use of augmented open addressing, and would instead require an entirely new approach to succinct hashing.

\subsection{Transforming a $k$-\kicktree into a $k$-kick hash table}

The $k$-\kicktree serves as the balls-to-slots scheme for all of the hash tables that we construct in this paper, but existing techniques for constructing the other parts of an augmented open-addressed hash table, (e.g.\ the query router, how to dynamically resize, etc.) are not themselves space and time efficient enough to fully take advantage of the efficiency of $k$-\kicktrees.  Next, we summarize the main technical obstacles that we overcome in order to use $k$-\kicktrees time and space efficiently in our hash tables.

\paragraph{An improved query router (Section \ref{sec:metadata}).}
We show how to build general-purpose query-routing data structures with strong space and time guarantees. Even if different keys have very different probe complexities from one another, our query-routing data structure uses space within a constant factor of optimal and supports constant-time queries/updates. The building blocks that we use to construct the query-router will likely also be useful in future work on related problems.

\paragraph{Supporting dynamic resizing (Section \ref{sec:variable}).} Past approaches~\cite{bender2021all, raman2003succinct, liu2020succinct} have performed resizing at a granularity of $1 + 1/\polylog n$ factors.\footnote{The specific ways in which resizing has been implemented have differed, with some papers \cite{raman2003succinct, liu2020succinct} performing resizing at a per-bin level, and others \cite{bender2021all} performing it globally.} This has required the data to be partitioned into $\polylog n$ chunks, and for the query-router to store an additional $\Theta(\log \log n)$ bits associated with each key in order to identify its chunk. In other words, dynamic resizing introduces yet another source of $\Theta(\log \log n)$ wasted bits per key. We give a general-purpose technique for avoiding this type of overhead---surprisingly, the technique results in the \emph{same} tradeoff curve that we encounter for probe-complexity: at the cost of $O(k)$ time per insertion/deletion, we can reduce the space overhead of resizing to $O(\log^{(k)} n)$ bits per key. 

\paragraph{Handling large keys/values (Section \ref{sec:largekeys}).}
Now consider the setting where the keys and values are $u$ and $v$ bits long, respectively, for some potentially large $u, v$ satisfying $u + v \le n^{o(1)}$. Past techniques have encountered several major obstacles in this case, resulting in the wasted space per key growing substantially as the key size $u$ becomes super-logarithmic \cite{arbitman2010backyard, liu2020succinct, raman2003succinct}. The only known succinct hash table that scales gracefully in the regime of $u + v = \omega(\log n)$ is the hash table of Raman and Rao \cite{raman2003succinct}, which achieves $O(\log (u + v))$ wasted bits per key with constant expected-time insertions---subsequent work \cite{arbitman2010backyard, liu2020succinct} on worst-case insertion times has encountered much larger space blowups due to technical difficulties surrounding the use of quotienting and the use of lookup-tables in hash tables with large keys. 

Our approach to handling large keys and values is to give a general-purpose reduction from the setting where $u \ge \omega(\log n)$ to the setting where $u = O(\log n)$. In essence, our reduction allows for us to move bits from the key length $u$ to the value-length $v$. 

We then show how to adapt our hash tables to support arbitrarily large values with \emph{no additional space wastage}. Here, we exploit a special property of the $k$-\kicktree, namely that it is capable of supporting a load factor of $1$, which ends up allowing for us to construct a dynamically-resized hash table in which there are \emph{no empty slots}. 

Combining these techniques, we conclude that the tradeoff curve in this paper is agnostic to key/value size: with $O(k)$-time per insertion/deletion, we can achieve $O(\log^{(k)} n)$ wasted bits per key.

\paragraph{Handling small keys/values (Section \ref{sec:small}).}
In addition to considering large keys, past work \cite{demaine2006dictionariis, raman2003succinct, liu2020succinct, arbitman2009amortized} has also focused in on the small case, where $u + v = \log n + t$ for some $t = o(\log n)$ (and where the universe $U = [2^{u}]$ of keys may have an arbitrarily small size satisfying $|U| = \omega(n)$).
In this setting, we show that if $t$ is even \emph{slightly} sublogarithmic, that is,
$$t = O(\log n / \log^{(k)} n)$$
for some positive constant $ k $, then it is possible to support constant-time insertions/deletions/queries while achieving $o(1)$ wasted bits per key. Prior to our work, this type of guarantee was only known to be possible in the much smaller regime of $t = \tilde{O}((\log n)^{1/ 3})$ \cite{arbitman2010backyard, raman2003succinct}. What makes our expanded range for $t$ interesting is that, as we shall see shortly, it enables us to design optimal dynamic filters for a large range of false-positive rates (in fact, for all false-positive rates except for those that are nearly polynomially small).

Our small-key result again follows from a general-purpose reduction, which in this case reduces the setting of small keys/values to the setting of larger keys/values. Interestingly, this reduction relies heavily on the ability to efficiently support dynamic-resizing and on the steep tradeoff curve between time/space for standard-sized keys/values.

\subsection{An application to optimal dynamic filters} 

Finally, in Section \ref{sec:filters}, we apply our small-keys result to the widely studied problem of maintaining space-efficient approximate-membership data structures, also known as filters. A (dynamic) \defn{filter} is a data structure that supports inserts/queries/deletes on a set of keys but that is permitted to return a false positive on a query with some probability $\epsilon$. Information theoretically, a filter must use at least $\mathcal{F}(n, \epsilon) = n \log \epsilon^{-1}$ bits \cite{carter1978exact}. 

It remains an open question what the optimal achievable wasted-bits-per-key is, that is, what is the smallest value of $r$ such that it is possible to construct a time-efficient dynamic filter using $\mathcal{F}(n, \epsilon) + nr$ bits. We remark that, here, $n$ is taken to be a fixed upper bound on the number of keys---if $n$ is permitted to change dynamically, then it is known that the optimal $r$ satisfies $r = \Omega(\log \log n)$ \cite{pagh2013approximate}. 

Filters tend to be used in applications where space efficiency is a central concern; the result is that most applications select $\epsilon$ such that $\log^{-1} \epsilon$ is very small (for a practical discussion of filters, see, e.g., \cite{fan2014cuckoo, bender2012don, dillinger2021ribbon}). This leads to the close relationship between the filter problem and the hash-table problem with small keys. 

Perhaps the most famous filter is the so-called Bloom filter \cite{bloom1970space}, which supports $O(\epsilon^{-1})$-time insertions and achieves $r = O(\log \epsilon^{-1})$ (the Bloom filter does not support deletions).  After a long line of work \cite{bloom1970space, carter1978exact, pagh2005optimal, bender2018bloom, bercea2020dynamic, liu2020succinct}, contemporary filters are able to achieve much stronger bounds than this. Indeed, there are now a number of filters \cite{pagh2005optimal, bender2018bloom, bercea2020dynamic} that  and that achieve
$$r = o(\log \epsilon^{-1})$$ wasted bits per key for all $\epsilon$ satisfying  $$\log \epsilon^{-1} \in [\omega(1), O(\log n)],$$
while supporting constant-time insertions/deletions/queries either in expectation \cite{pagh2005optimal} or in the worst case \cite{bercea2020dynamic, bender2021all} (with high probability).

The central open question in the study of filters is whether it is possible to achieve $r = O(1)$ wasted-bits-per-key for all $\epsilon$. It is known that $\Omega(1)$ wasted-bits-per-key are \emph{necessary}, at least for some values of $\epsilon$ \cite{lovett2010lower} (namely, $\epsilon = \Theta(1)$), but it is not known whether $O(1)$ wasted-bits-per-key is \emph{achievable}. 

We show that, for any positive constant $k$, it is possible to achieve a filter that uses space 
\begin{equation}
r = \log e + o(1) = O(1)
    \label{eq:const}
\end{equation} wasted bits per key for all inverse-power-of-two $\epsilon$ satisfying 
$$\log \epsilon^{-1} \in [\omega(1), \log n / \log^{(k)} n],$$
while supporting worst-case constant-time insertions/deletions/queries (with high probability). The total space used by the data structure is therefore
$$n \log \epsilon^{-1} + n\log e + o(n)$$
bits. This resolves the question of whether $r = O(1)$ is achievable in all cases except for when $\log^{-1} \epsilon$ is very close to $\log n$. Finally, we show that for any value of $\log^{-1} \epsilon$ (including $\log^{1} \epsilon = \Theta(\log n)$), the same time/space tradeoff curve that we achieve for hash tables, in which $O(k)$-time insertions/deletions yield $O(\log^{(k)} n)$ wasted bits per key, is achievable for dynamic filters.

We remark that the specific constant $\log e$ that we achieve in Equation \eqref{eq:const} is information-theoretically optimal for any filter that is constructed by storing fingerprints in a hash table. (This includes all modern dynamic filters \cite{carter1978exact, pagh2005optimal, bender2018bloom, bercea2020dynamic, liu2020succinct}.) Thus, improving upon this constant would require a fundamentally new approach to building constant-time filters. We conjecture that no such improvements are possible (even for non-constant-time dynamic filters)---proving a lower bound for this claim is an appealing direction for future work. 

\subsection{Preliminaries}\label{sec:prelim}

We conclude the section by formalizing several preliminaries that we will need throughout the paper.

\paragraph{Notation.}
We use $[i, j]$ to denote the range $\{i, \ldots, j\}$, we use $[i]$ to denote $[1, i]$, and we use $\log^{(i)} n$ to denote the function given by $\log^{(0)} n = n$ and $\log^{(i)} n = \max(\log \log^{(i)} n, 1)$ for all $i \ge 0$. Note that, as a matter of convention, we do not allow $\log^{(i)} n$ to become sub-constant.

\paragraph{High-probability guarantees.} We say that an event occurs with high probability (w.h.p.) in $n$ if it occurs with probability $1 - 1 / \poly(n)$. Our hash tables will offer a deterministic guarantee on the running times of queries, a high-probability guarantee on the running time of any given insertion/deletion, and a high-probability guarantee on the space consumption at any given point in time. To simplify discussion, we will allow for our hash tables to have alternative failure modes (e.g., some bin overflows), with the implicit assumption that whenever a low-probability failure event occurs during an insertion/deletion, the hash table is then be rebuilt from scratch using new randomness---this means that failure events cause the hash table to violate time/space guarantees, but not correctness guarantees.

\paragraph{Simulating fully random hash functions.}
Whereas early work on hash tables \cite{Fredman82FKS, dietzfelbinger1990new, dietzfelbinger1988dynamic} was bottlenecked by the known families of hash functions, there are now well-established techniques \cite{siegel1989universal, pagh2003uniform, dietzfelbinger2003almost, arbitman2009amortized, bender2021all} for simulating fully random hash functions in hash tables. Notably, Siegel \cite{siegel1989universal} showed that for some positive $\epsilon > 0$, there is a family of constant-time hash functions that can be constructed in time $o(n)$ and that is $n^{\epsilon}$-independent.\footnote{Siegel's construction requires that the universe $U$ of keys has at most polynomial size---but it can also be used with a larger universe by first performing dimension reduction to a $\poly(n)$-size universe using a pairwise independent hash function.} In the context of hash tables, this can be amplified to simulate $\poly(n)$-independence \cite{arbitman2009amortized, bender2021all} with the following ``sharding'' technique: use a hash function $h_1$ to partition the elements into buckets with sizes in the range $[n^{\delta}, n^{\delta} + n^{2\delta/3}]$; then implement each bucket as its own hash table, where the all of the buckets share access to a single $n^{\epsilon}$-independent family $\mathcal{H}$ of hash functions---if a given bucket has size $m = \Theta(n^{\delta})$, then $\mathcal{H}$ is $\poly(m)$-independent. Thus we can assume without loss of generality that we have access to $\poly(n)$-independent hash functions.

In Section \ref{sec:quotient}, in order to perform quotienting, we will also want access to random \emph{permutation} hash functions, that is, hash functions $h$ that are bijective on some universe $U$ of keys. As long as $|U| \le \poly(n)$, then there are again well-established techniques for simulating full randomness. 
Naor and Reingold (Corollary 8.1 of \cite{naor1999construction}), building on seminal work by Luby and Rackoff \cite{luby1988construct}, showed how to construct in time $o(n)$ an $n^\epsilon$-wise $1 / n^\delta$-dependent family of permutations with constant-time evaluation. Subsequent work showed how to amplify this to $n^\epsilon$-wise $1/\poly(n)$-dependence, which allows for the simulation of $n^\epsilon$-wise independence with probability $1 - 1 / \poly(n)$. Finally, using a similar sharding technique as described above (but with $h_1$ implemented using a single-round Feistel permutation, as in \cite{arbitman2009amortized}), one can use such a family of hash functions to simulate $\poly(n)$-independence in a hash table (see Section 7 of \cite{bender2021all} for an in-depth discussion). Thus, as long as $|U| \le \poly(n)$, then we can assume without loss of generality that we have access to $\poly(n)$-independent permutation hash functions.

\paragraph{Machine model.}
Our analyses will be in the standard word RAM model. If keys/value pairs are each $w$ bits long, then we shall assume a machine word of size at least $w$. To analyze space consumption, we will assume that algorithms have the ability to allocate/free memory with $O(\log n)$-bit pointers. We remark, however, that all of our algorithms have highly predictable allocation patterns, and are therefore straightforward to implement using a small number of large memory slabs (e.g., when keys/values are $\Theta(\log n)$ bits, we need only to allocate $\polylog(n)$ slabs of memory at a time).

\section{The Probe-Complexity Problem}\label{sec:probe}

Recall that the probe-complexity problem can be defined formally as follows. Let $ U $ be a universe of balls, and let $n \in \mathbb{N}$ be the number of slots\footnote{Whereas the hash table literature typically uses $n$ to be the number of keys/values, the balls-to-bins literature typically uses $n$ to be the number of slots (or bins). In this section, we follow the balls-to-bins convention, and in the rest of the paper we follow the hash-table convention.}, and let $\epsilon \in [0, 1)$ be a load-factor parameter. A \defn{balls-to-slots scheme} assigns to every ball $ x\in U$ a fixed \defn{probe sequence} $h(x) = \langle  h_1 (x), h_2 (x),\ldots\rangle$, each $h_i(x) \in [n]$.

We define the \defn{probe-complexity problem} as follows. There are $ n $ \defn{slots} each with capacity 1. An oblivious adversary (who does not know $h$) selects a sequence of ball insertions/deletions such that at most $(1 - \epsilon)n + 1$ balls are present at a time. A balls-to-slots scheme must maintain an assignment of balls to slots such that each ball $ x $ (that is present) is assigned to slot $h_i(x)$ for some $i$. If a ball $x$ is in slot $r$, then we say that $x$ has \defn{probe complexity} $\Theta(1 + \log i)$ where $i = \operatorname{argmin}_j \{h_j(x) = r\}$.

To simplify discussion, we shall also give the balls-to-slots scheme $n$ extra \defn{special slots}. Any ball that is stored in a special slot automatically has probe complexity $\log n$. Whenever a ball insertion occurs, the ball is first placed into a special slot. The balls-to-slots scheme can then move that ball (and other balls) around in order to reduce the average probe complexity of the balls that are present. 

The balls-to-slots scheme is measured by two objectives: the average probe complexity of the balls that are present; and the \defn{switching cost}, which is the number of balls that the balls-to-slots scheme moves around on any given insertion/deletion. When a balls-to-slots scheme is used in an augmented open-addressed hash table, the switching cost is (a lower bound on) the time spent on a given insertion/deletion, and the probe complexity of a key $x$ is (a lower bound on) the number of metadata bits that must be stored to support constant-time queries for $x$.  

In this section, we give an optimal solution to the probe-complexity problem (Subsection \ref{sec:probeupper}), achieving probe-complexity $O(\log^{(k)} n)$ with switching cost $O(k)$. This holds even when $\epsilon = 1/n$, meaning that there are up to $ n $ balls present at a time (and there are up to $n - 1$ balls present prior to any given insertion).

We then also prove a matching lower bound in Subsection \ref{sec:lower}: any balls-to-slots scheme that supports $\epsilon \le 1 / \log^{(O(1))}(n)$ with expected average probe complexity $O(\log^{(k)} n)$ must incur average switching cost $\Omega(k)$. Note that, whereas our upper bound supports $\epsilon = 1/n$ (i.e., the slots are completely full), our lower bound allows for $\epsilon$ to be as large as $1 / \log^{(O(1))} n$, without changing the answer for what the optimal tradeoff curve between probe complexity and switching cost is.

\subsection{A balls-to-slots scheme with small average probe complexity}\label{sec:probeupper}

In this section, we fix $\epsilon = 1/n$, and we construct a balls-to-slots scheme that achieves switching cost $k + 1$ ($1$ for inserting a ball, and $k$ for moving around balls already in the system) while also achieving expected average probe complexity $\Theta(\log^{(k + 1)} n)$. At the end of the section, we also show how to transform the bound on average probe complexity into a high-probability result.

\paragraph{Defining each ball's probe sequence.} 
Define $s_0 = n$ and define $ s_i = \Theta((\log ^ { (i) } n) ^ 6)$ to be a power of two for each $i \in [1, k]$. We shall assume for simplicity that $n$ is divisible by $s_i$ for each $i > 0$, but the same arguments easily extend to arbitrary $n$.

We shall consider $k + 1$ different ways of partitioning the $n$ slots into bins: for $i \in [0, k]$, the \defn{depth-$ i $ partition} breaks the slots into contiguous bins of size $s_i$. For each depth-$i$ bin $b$, with $i > 0$, the \defn{parent bin} $b'$ of $b$ is the depth-$ (i -1) $ bin that contains $b$. (And $b$ is a \defn{child} of $b'$.) So the partitions are arranged in a tree, where the depth-$i$ components are children of the depth-$(i - 1)$ components, and where the branching factor of the tree decreases roughly exponentially between levels.  
We call this tree the \defn{$k$-\kicktree}.

Before defining $h$, we define an auxiliary function $g$. Each ball $x$ randomly selects a leaf of the tree (i.e, some depth-$k$ bin $b$) and defines $g_i(x)$ to be the depth-$i$ ancestor of $b$. In other words, each $g_i(x)$ is a uniformly random depth-$i$ bin, and the sequence $g_0(x), g_1(x), \ldots, g_k(x)$ forms a root-to-leaf path through the \kicktree.

The function $h_i(x)$ first cycles through the slots of $g_{k}(x)$, then the slots of $g_{k - 1}(x)$, then the slots of $g_{k - 2}(x)$, etc. Formally, this means that for each depth $i$ and for each $j \in [s_i]$, $h_{(k + 1 - i) s_i + j}(x)$ is the $j$-th position in bin $g_i(x)$.  Since the bin $g_0(x)$ contains \emph{all} slots in $[n]$, the sequence $\{h_i(x)\}_{i \in [(k + 1) n + 1, (k + 2) n ]}$ hits every slot, so we do not need to define $h_i$ for $i > (k + 2) n$.

Whenever a ball $ x $ is inserted, it ends up at some depth $i$, and within that depth it is assigned to some position $j \in [s_i]$ of bin $g_i(x)$. The ball's probe complexity is then
$$O(1 + \log (k + 1 - i) + \log s_i).$$
Since $k + 1 - i = O(\log^* s_i)$, the probe complexity reduces to
$$O(\log s_i).$$ 
Throughout the rest of the section, we will think of each ball $x$'s position as being determined by a pair $(i, j)$, where $i$ is a depth and $j$ is a position in $g_i(x)$, rather than being determined directly by the probe sequence $h$. If a ball $x$ is associated with depth $i$, we will treat it as having probe complexity $\Theta(s_i)$.

\paragraph{The structure of a ball insertion.}
Call a depth-$ i $ bin \defn{saturated} if the bin contains no free slots and if all of the balls in the bin are associated with depths $i$ or greater. Note that, when we are performing an insertion, $g_0(x)$ cannot be saturated, but $g_i(x)$ for $i > 0$ may be.

Whenever a ball $x$ is inserted, we select a depth $i$ such that none of the bins $g_0(x), g_1(x), \ldots, g_i(x)$ are saturated. (We will describe the process for selecting $i$ later.) We assign $ x $ to bin $g_i(x)$ with depth $ i $. If there is a free slot in $g_i(x)$, then we use it; otherwise, since $g_i(x)$ is not saturated, the bin must contain a ball $x'$ associated with some depth $i' < i$. We assign $ x $ to the slot that $x'$ is in, and we reassign $x'$ to a new slot as follows. If there is a free slot in $g_{i'}(x')$, then we use it; otherwise, since $g_{i'}(x') = g_{i'}(x)$ is not saturated, the bin must contain a ball $x''$ associated with some depth $i'' < i'$. We assign $x'$ to the slot that $x''$ is in, and we reassign $x''$ to a new slot, etc., where $x''$ may displace some ball $x'''$ at a depth $i''' < i''$, and so on. In effect, we treat the depths as priorities, so that whenever a ball $y$ is moved, it is permitted to displace any other ball $y'$ that is of a lower priority.

Each insertion has switching cost at most $k + 1$, since it places the ball that is being inserted and then rearranges at most one ball in each depth $\{0, 1, \ldots, k - 1\}$. Moreover, whenever a ball is moved, the depth that it is in stays the same, and thus the $O(\log s_i)$-bound on the probe complexity for that ball also stays the same. In order to achieve $O(\log^{(k + 1)} n)$ average probe complexity (in expectation), it therefore suffices to ensure that, whenever a ball is inserted, the expected probe complexity for the new ball is $O(\log s_k) = O(\log^{(k + 1)} n)$.

\paragraph{Choosing which depth to use.}
The final piece of the algorithm that we must specify is how to choose the depth $ i $ that a given ball insertion will use. 

The most natural approach is to be greedy: select the largest $ i $ such that none of bins $g_0(x), g_1(x), \ldots, g_i(x)$ are saturated. This optimizes the probe complexity of the current insertion but comes with a downside. We are not doing anything to control which bins are saturated, so even though we are selecting $ i $ greedily, we cannot argue that $i$ will actually be large for any given insertion (for example, what if $g_1(x)$ is saturated?).

Our solution is to take an \emph{almost} greedy approach. Each ball $ x $ is assigned an independent hash $ s (x) \in [0, k]$ satisfying
$$\Pr[s(x) < i] = 1 / (\log^{(i)} n)^2$$
for each $i \in [1, k]$. The hash $s(x)$ dictates the  \emph{maximum possible depth} that ball $x$ is permitted to be in.  Each insertion $x$ uses depth $\min(j, s(x))$, where $j$ is the largest value such that none of the bins $g_0(x), g_1(x), \ldots, g_j(x)$ are saturated.

\paragraph{Analyzing a given insertion.} We now analyze the expected probe complexity of a given ball. 

\begin{lemma}
Consider the insertion of some ball $ x $ into a $k$-\kicktree with $n$ slots. The expected probe complexity of $ x $ is $ O (\log^{(k+1)} n) $.
\label{lem:insertionprobe}
\end{lemma}
\begin{proof}
Let $j$ be the largest value such that none of the bins $g_0(x), g_1(x), \ldots, g_j(x)$ are saturated. Then the probe complexity of $x$, after being inserted, is 
$$O(\log \max(s_{s(x)}, s_j)) = O(\log s_{s(x)}) + O(\log s_j).$$

We can bound the expected value of the first quantity by
\begin{align*}
\E[\log s_{s(x)}] & = \log s_k + \sum_{i \in [0, k)} \Pr[s(x) = i] \cdot \log s_i \\
                  & \le  O(\log^{(k+1)} n) + \sum_{i \in [0, k)} \Pr[s(x) < i  + 1] \cdot \log s_i \\
                  & = O(\log^{(k+1)} n) + \sum_{i \in [0, k)} \frac{1}{(\log^{(i + 1)} n)^2} \cdot \log s_i \\
                  & = O(\log^{(k+1)} n) + \sum_{i \in [0, k)} \frac{1}{(\log^{(i + 1)} n)^2 } \cdot \log (\log^{(i)} n)^6 \\
                  & = O\left(\log^{(k+1)} n + \sum_{i \in [0, k)} \frac{1}{(\log^{(i + 1)} n)^2 } \cdot \log^{(i + 1)} n\right) \\
                  & = O\left(\log^{(k+1)} n + \sum_{i \in [0, k)} \frac{1}{\log^{(i + 1)} n}\right) \\
                  & = O(\log^{(k+1)} n).
\end{align*}

We can bound the expected value of the second quantity by
\begin{equation}
\E[\log s_{j}]  \le \log s_k + \sum_{i \in [0, k)} \Pr[g_{i + 1}(x) \text{ saturated}] \cdot  \log s_i. 
\label{eq:sj}
\end{equation}
In order for $g_{i + 1}(x)$ to be saturated (prior to $x$'s insertion), there must be $s_{i + 1}$ balls $y$ present that satisfy $g_{i + 1}(y) = g_{i + 1}(x)$ and $s(y) \ge i + 1$. For a given $y \neq x$, 
\begin{align*}
& \Pr[g_{i + 1}(y) = g_{i + 1}(x) \text{ and } s(y) \ge i + 1] \\
& = \Pr[g_{i + 1}(y) = g_{i + 1}(x)] \cdot \Pr[s(y) \ge i + 1] \\
& = \frac{1}{n / s_{i + 1}} \cdot \left(1 - \Pr[s(y) < i + 1]\right) \\
& = \frac{1}{n / s_{i + 1}} \cdot \left(1 - 1 / (\log^{(i + 1)} n)^2\right). \\
\end{align*}
The number $Y$ of such $y$ therefore satisfies
\begin{align*}
\E[Y] & \le  n \cdot \frac{1}{n / s_{i + 1}} \cdot \left(1 - 1 / (\log^{(i + 1)} n)^2\right) \\
& = s_{i + 1} \cdot \left(1 - 1 / (\log^{(i + 1)} n)^2\right) \\
& = (\log^{(i + 1)} n)^6 - (\log^{(i + 1)} n)^4. \\
\end{align*}
Since $Y$ is a sum of independent indicator random variables, we can apply a Chernoff bound to deduce that
\begin{align*}
\Pr[Y \ge (\log^{(i + 1)} n)^6] & \le e^{-\Omega(\log^{(i + 1)} n)^2} \\
                                & \le O\left(\frac{1}{(\log^{(i)} n)^2}\right). 
\end{align*}
This is an upper bound on the probability that $g_{(i + 1)}(x)$ is saturated. Thus, by \eqref{eq:sj}, 
\begin{align*}
\E[\log s_{j}] & \le \log s_k + \sum_{i \in [0, k)} \frac{1}{(\log^{(i)} n)^2} \cdot \log s_i \\
& = O\left(\log^{(k+1)} n + \sum_{i \in [0, k)} \frac{1}{(\log^{(i)} n)^2} \cdot \log^{(i + 1)} n\right) \\
& = O(\log^{(k+1)} n).\\
\end{align*}
This completes the proof of the lemma.
\end{proof}

It's worth taking a moment to understand the bottlenecks in the Lemma~\ref{lem:insertionprobe}. For convenience, let us focus on the setting where we are aiming for average probe complexity $O(1)$, so $k = \Theta(\log^* n)$; and further assume that, if a ball is placed in a depth-$d$ bin, then the ball has probe complexity $\Theta(\log s_d)$. Consider some bin $ b $ with depth $ i > 0$, meaning that the bin has size $ s_i $. If we want to ensure that the probability of $b$ being saturated is $o(1)$, then we must ensure that the expected number of elements in $b$ is $s_i - \omega(\sqrt{s_i})$ (because the standard deviation of the number of elements in the bin is $\Theta(\sqrt{s_i})$). This means that the hash function $s(x)$ must satisfy 
$$\Pr[s(x) < i] = \omega(1/\sqrt{s_i}).$$
However, whenever $s(x) < i$, the probe complexity of $x$ is forced to be at least $\log s_{i - 1}$. Thus the expected probe complexity of each ball $x$ must be at least
$$\omega\left(\frac{\log s_{i - 1}}{\sqrt{s_i}}\right).$$
Since we want average probe complexity $O(1)$, it follows that $\frac{\log s_{i - 1}}{\sqrt{s_i}} \le 1$, or equivalently,
$$s_{i - 1} \le 2^{\sqrt{s_i}}.$$
This is the inequality that fundamentally limits the rate at which the $s_i$'s can shrink and that forces us to have $\Omega(\log^* n)$ depths in order to achieve average probe complexity $O(1)$. In fact, we'll see in Section \ref{sec:lower} that this relationship between probe complexity and switching cost is fundamental---no balls-to-slots scheme can do better than the $k$-\kicktree does.

An immediate consequence of Lemma \ref{lem:insertionprobe} is:
\begin{theorem}
For any $k \in [\log^* n - 1]$, the $k$-\kicktree is a balls-to-slots scheme with $\epsilon = 1/n$ that achieves worst-case switching cost $k + 1$ and expected average probe complexity $O(\log^{(k+1)} n)$. 
\label{thm:probeupper}
\end{theorem}

We conclude the section by transforming our bound unexpected average probe complexity into a high-probability bound.

\begin{theorem}
For any $k \in [\log^* n - 1]$, there is a balls-to-slots scheme with $\epsilon = 1/n$ that achieves worst-case switching cost $k + 1$ and average probe complexity $O(\log^{(k+1)} n)$, with probability $1 - 1 / 2^{n^{\Omega(1)}}$ at any given moment. 
\label{thm:probeupper2}
\end{theorem}
\begin{proof}
Let $\rho:U \rightarrow [\log n]$ be a fully independent and uniformly random hash function. Whenever a ball $x$ is inserted, if $\rho(x) = 1$, then place $x$ into a special slot.\footnote{Alternatively we could place $x$ into whatever slot $s$ is free and then move $x$ to a different free slot whenever that slot $s$ needs to be used by a different ball. This would increase the switching cost by at most $1$ per operation and avoid the use of special slots.} With probability $1 - 1 / 2^{n / \log n}$, the number of balls in special slots is  $O(n / \log n)$ at any given moment, meaning that they contribute $O(1)$ to the average probe complexity.

We hash the remaining balls $x$ (i.e., balls $x$ satisfying $\rho(x) > 1$) randomly to subarrays of size $\sqrt{n}$. The expected number of balls in a given sub-array is at most $\sqrt{n} (1 - 1 / \log n)$, so with probability $1 - 1 / 2^{n^{\Omega(1)}}$ each of the subarrays receives at most $\sqrt{n} - 1$ balls at any given moment. In the rare case that a sub-array overflows, we can simply put the ball $x$ being inserted into a special slot.

Finally, we implement each subarray using a $k$-\kicktree. By Theorem \ref{thm:probeupper}, each insertion incurs switching cost $k + 1$ and each subarray independently incurs expected average probe complexity $O(\log^{(k+1)} n)$. 

If we define $X_1, X_2, \ldots, X_{\sqrt{n}}$ to be the average probe complexities of the subarrays, then $\{X_i\}$ are independent random variables satisfying $X_i \le O(\log n)$ and $\E[X_i] = O(\log^{(k+1)} n)$. By a Chernoff bound, we have with probability $1 - 1 / 2^{n^{\Omega(1)}}$ that the average probe complexity across the entire balls-to-slots scheme is $O(\log^{(k+1)} n)$, as desired.
\end{proof}

\subsection{A lower-bound on the relationship between switching cost and probe complexity}\label{sec:lower}

In this section we will construct a sequence of insertions/deletions such that, in order for an online balls-to-slots scheme to achieve small average probe complexity, they must incur a large average switching cost. Since we are constructing a lower bound, we shall refer to the balls-to-slots scheme that we are analyzing as our \defn{adversary}.

Let $U$ be the universe of balls, let $h$ be the function mapping each ball $x$ to a probe sequence $\{h_i(x)\}$, and let $n$ be the number of slots. For $ i\in\N, j\in [n] $, define 
$$q(h, i, j) = n\Pr_{ x\in U } [h_k (x) = j \text{ for some }k \le i].$$
Intuitively, if there are $n$ random balls present, then $q(h, i, j)$ represents the expected number of balls that are capable of residing in slot $j$ with probe complexity at most $1 + \log i$.

If the $h_i(x)$'s are selected uniformly and independently in $[n]$, then we will have $q(h, i, j) = \Theta(i)$ for each $i \in [n]$. Call $h$ \defn{nearly uniform} if $ q (h, i, j) < \poly (i) $ for all $ i, j $. As a minor technical convention, we will also allow for $h_i(x)$ to be null, in which case it does not contribute to any $q(h, i, j)$ (and $h_i(x)$ cannot be used by any ball assignment). 

We shall begin by proving a lower bound that holds assuming a nearly uniform $h$. We shall also initially assume that $\epsilon = 1/ n$ and that the average probe complexity being achieved by the balls-to-slots scheme is $O(1)$. We will then remove these assumptions at the end of the section.

\begin{theorem}
 Suppose the universe $U$ has sufficiently large polynomial size. Consider any balls-to-slots scheme that uses nearly-uniform probe sequences, that achieves expected average probe complexity $O(1)$ (across all balls in the system at any given moment), and that supports $\epsilon = 1/n$. The expected amortized switching cost per insertion/deletion must be $\Omega(\log^* n)$.  
\label{thm:switchinglowerbound}
\end{theorem}

Throughout the rest of the section, we shall consider an input sequence that begins right after $n - 1$ random balls have just been inserted, and then proceeds to perform $M = \poly(n)$ insertions and the same number of deletions. The insertions and deletions alternate; each insertion inserts a random ball (which with probability $1 - 1 / \poly(n)$ has never been inserted in the past); and each deletion deletes a random ball out of those present. 

Define $ L = \lceil (\log^* n) / 2 \rceil$. Define $\Tower (0) = 1$ and $\Tower(i) = 2^{\Tower(i - 1)}$ for all integer $i > 0$. Say that a ball $x$ is in \defn{level 0} if it has been assigned to a slot $h_i(x)$ for some $i \le \Tower(L)$, and say that a ball $x$ is in \defn{level $j \in \{1, 2, \ldots, L\}$} if it has been assigned to a slot $h_i(x)$ for some $i \in (\Tower(L + j - 1), \Tower(L + j)]$. If a ball is in a special slot, or if it has been assigned to a slot $h_i(x)$ with $i \ge n$, then the ball is said to be in level $L$.

Say that a move by the adversary has \defn{impact} $r$ if it decreases the level of some ball by $r$. Positive impact means that the ball's level decreased, and negative impact means that the ball's level increased.

\begin{lemma}
 For $i \in [M]$, define $\alpha_i$ to be the sum of the impacts of the moves that the adversary performs during the $i$-th insertion, and define $\beta_i$ to be the sum of the impacts of the moves that the adversary performs during the $i$-th deletion. Finally, define $\psi = \sum_{i \in [M]} (\alpha_i + \beta_i)$ to be the total impact by the adversary across all insertions/deletions. Then
$$\E[\psi] = \Theta(ML).$$
\label{lem:impact}
\end{lemma}
\begin{proof}
Recall that $M$ is the number of insertions (resp. deletions) performed, and that $L$ is the number of levels. Define a dynamically-changing quantity $ J $ to be the sum of the levels of the balls in the system at any given moment. 

Each insertion places a ball into a special slot, thereby increasing $ J $ by $L$. On the other hand, we claim that each deletion decreases $J$ by $O(1)$ in expectation. To see this, observe that the deletion decreases $J$ by $s$ where $s$ is the level of the element being deleted. Since the adversary guarantees an expected average probe complexity of $O(1)$, we have that $\E[s] = O(1)$, which means that $J$ decreases by $O(1)$ in expectation.

By the definitions of $\alpha_i$ and $\beta_i$, we have that during the $i$-th insertion (resp. $i$-th deletion), the adversary's moves decrease $J$ by $\alpha_i$ (resp. $\beta_i)$. Across all operations, the total effect of the adversary's moves on $J$ is to decrease it by $\psi$. 
If $J_0$ is the value of $J$ prior to the first of the $2M$ operations and $J_*$ is the value of $J$ after the final operation, then 
$$\E[J_*] = \E[J_0] + LM - \Theta(M) - \E[\psi] = \E[J_0] + \Theta(LM) - \E[\psi],$$
where the $LM$ term accounts for insertions, the $M$ term accounts for deletions, and the $\psi$ term accounts for moves by the adversary. On the other hand, $J_0$ and $J_*$ are both deterministically in the range $[0, O(n\log^* n)]$, so we must have
$$\Theta(LM) - \E[\psi] \le O(n \log^* n).$$
Since $M$ is a large polynomial, it follows that $\E[\psi] = \Theta(LM)$,  as desired.
\end{proof}

The main technical ingredient to complete the proof of Theorem \ref{thm:switchinglowerbound} will be to construct a potential function $\phi $ with the following properties:
\begin{itemize}
\item \textbf{Property 1:} Each insertion/deletion increases $\phi $ by at most $O(1)$ in expectation.
\item \textbf{Property 2:} If a move by the adversary has impact $ r \in \mathbb{Z}$, it decreases $\phi$ by $r \pm O(1)$. 
\item \textbf{Property 3:} $\phi$ always satisfies $0 \le \phi \le nL$.
\end{itemize}

Before we construct $\phi $, let us assume the existence of such a $\phi$ and use it to complete the proof. At any given moment, define $\psi$ to be the sum of the impacts of the moves that the adversary has made so far. We will examine how the quantity $\psi + \phi$ evolves over time.

By Property 3, the quantity $\psi + \phi$ is initially at most $nL$. By Property 1, each insertion/deletion increases $\psi + \phi$ by $0 + O(1) = O(1)$ in expectation. By Property 2, each move by the adversary increases $\psi + \phi$ by at most $r - (r - O(1)) = O(1)$ (deterministically). Thus, after $M$ insertions/deletions have been performed, if $ k $ is the total number of moves that the adversary makes, then
$$\E[\psi + \phi] \le nL + O(M) + O(\E[k]) = O(M) + O(\E[k]).$$
On the other hand, by Property 3, $\E[\psi] \le \E[\psi + \phi],$ so
$$\E[\psi] \le O(M) + O(\E[k]).$$
Lemma \ref{lem:impact} tells us that $\E[\psi] = \Theta(ML)$. Thus
$$ML \le O(M) + O(\E[k]),$$
which means that $\E[k] = \Omega(ML) = \Omega(M \log^* n)$, hence Theorem \ref{thm:switchinglowerbound}. The main challenge is therefore to construct a potential function $\phi$ with the three desired properties.

\paragraph{Constructing the potential function $\phi$.}
The basic idea behind $\phi$ is that it should approximate the amount of impact that the adversary could hope to achieve with a small number of moves. One way to do this would be as follows. We could define $\mathcal{S}$ to be the set of all possible move sequences that the adversary could make; for each $S \in \mathcal{S}$, we could define $I(S)$ to be the total impact of $S$ and $|S|$ to be the number of moves in $S$; and we could define
$$\phi = \max_{S \in \mathcal{S}} \left(I(S) - c |S|\right)$$
for some large positive constant $c$. This potential function would exactly capture the adversary's ability to achieve large impact with a small number of moves, but it comes with the drawback that it can behave somewhat erratically with respect to insertions, deletions, and adversary-moves.

A key idea in this section is to construct $\phi$ in a more intricate way, still upper-bounding the amount of impact that the adversary can achieve cheaply, but while also intentionally designing $\phi$ to behave nicely. In order to give the technical definition of $\phi$, we must first define the notion of an \defn{$i$-stanza}, which intuitively corresponds to a sequence of moves in which the adversary is able to reduce the level of some ball $b$ from $\ge i$ to $\le i - 3$ while preserving for every other ball $b'$ how the level $\ell'$ of $b'$ compares to the quantities $i - 2, i - 1, i$. 

Define the  \defn{level of a slot} $ s $ to be the level of the ball in the slot, if there is such a ball, and to be $L$ otherwise.
For $i \in [L]$, define an \defn{$i$-stanza} to be a sequence of slots $s_1, \ldots, s_j$ such that slots $s_1$ and $s_j$ have levels at least $i$; such that slots $s_2, \ldots, s_{j - 1}$ have levels at most $i - 3$; such that each slot $s_k$, $k \in [j - 1]$, contains a ball $x$ that can be placed into $s_{k + 1}$ with a new level of at most $i - 3$; and such that $s_2, \ldots, s_{j - 1}$ are distinct. Note that, by design, $s_2, \ldots, s_j$ cannot be special slots (since they must be capable of containing a ball with level $\le i - 3$), but $s_1$ can be.

Importantly, the final slot $s_j$ in a stanza does \emph{not} have to be an empty slot in order for the stanza to be valid. With that said, if the final slot $s_j$ were empty, then the stanza would have a very intuitive interpretation: one could think of the stanza is representing a possible chain of ball moves, where the first ball move (from slot $s_1$ to slot $s_2$) decreases the level of some ball from $\ge i$ to $\le i - 3$,  where each subsequent ball move (from slot $s_k$ to slot $s_{k + 1}$ for some $k > 0$) maintains the level of some ball to be at most $i - 3$, and where the final move places a ball into an empty slot. 

We say that an $i$-stanza $s_1, \ldots, s_j$ has \defn{size} $j$ and has \defn{potential} $1 - (j - 1) / L$. We refer to $s_1$ as the \defn{starting slot} of the stanza, to $s_2, \ldots, s_{j - 1}$ as the  \defn{internal slots} of the stanza, and to $s_j$ as the \defn{final slot} of the stanza. We say that a collection of $i$-stanzas are \defn{disjoint} if each slot with level $\ge i$ is used at most once as a starting slot and at most once as a final slot, and if each slot with level $\le i - 3$ is used at most once as an internal slot. (The only overlap allowed is that the starting slot of one stanza may be the ending slot of another.) The \defn{potential} of a disjoint collection of $i$-stanzas is the sum of the potentials of the individual stanzas.

For $i \in [L]$, define $\phi_i$ to be the maximum potential of any disjoint collection of $ i $-stanzas. Finally, define the potential function $\phi$ by 
$$\phi = \sum_{i = 3}^L \phi_i.$$ 

\paragraph{The intuition behind $\phi$.}
Before analyzing $\phi$, let us give a bit more intuition for why $\phi$ acts as a natural upper-bound for how much impact the adversary can achieve cheaply  (i.e., with only a small number of moves relative to the impact being achieved). 

Consider any sequence of moves that the adversary could perform, and define a \defn{realized stanza} to be a sequence of slots $s_1, \ldots, s_j$ such that $s_j$ is an empty slot and, for each $k \in [j - 1]$, the ball from slot $s_k$ gets moved to the next slot  $s_{k + 1}$ in the sequence. One can think of a realized stanza as a sequence of moves, where balls $x_1, \ldots, x_{j - 1}$ are in slots $s_1, \ldots, s_{j - 1}$ and are being moved to positions $s_2, \ldots, s_j$, respectively. Each ball $x_k$ is moved from some initial level $b_k$ to some potentially different level $e_k$. (As an edge case,since there is no ball initially in $s_j$, define $b_j = L$, and leave $e_j$ undefined.) 

For a given move, from some level $b_k$ to some level $e_k$, there are three cases for the adversary. If $e_k \in \{b_k - 2, b_k - 1\}$, then we think of the move as having been neither good nor bad for the adversary---the move created $\Theta(1)$ impact, but at the cost of $1$ move. If $e_k \le b_k - 3$, then we think of the move as being good for the adversary, and we say that the adversary has \defn{stolen} $b_k - e_k - 2$ levels $b_k, b_k - 1, \ldots, e_k + 3$. Finally, if $e_k \ge b_k$, then we think of the move as being bad for the adversary, and we say that the adversary has \defn{paid} for $e_k - b_k + 1$ levels $b_k + 1, \ldots, e_k + 2$. Whenever the adversary pays for a level $i$, that cancels out the previous time that the adversary stole that level $i$. In order for the adversary to steal a level $i$ without subsequently paying for it, there must be a sequence $(b_k, e_k), \ldots, (b_{k'}, e_{k'})$ such that $b_k \ge i$, such that $e_k, b_{k + 1}, e_{k + 1}, b_{k + 2}, \ldots, e_{k' - 1} \le i - 3$, and such that $b_{k'} > i$. This sequence of ball moves corresponds exactly to an $i$-stanza. In other words, each $i$-stanza represents a possible opportunity for the adversary to steal level $i$ without subsequently paying for it.

In order for a realized stanza to be worthwhile to the adversary, however, the adversary must perform an average of $\omega(1)$ steals per move. This means that, on average, each level of the $L$ levels $i>0$ must be stolen $\omega(1)$ times for every $L$ moves that are performed. In other words, whenever the adversary steals level $i$, but then fails to steal level $ i $ again for $L$ moves, then that first steal wasn't actually worthwhile. The value of a given steal can be modeled as $1 - q / L$, where $q$ is the number of subsequent moves until the next steal of the same level. This is why we define the potential of an $i$-stanza in the way that we do: the longer that a $i$-stanza is, the less worthwhile of an opportunity that it represents for the adversary.

In summary, each $i$-stanza represents an opportunity for the adversary to steal level $i$ without subsequently paying for it; and the  $i$-stanza's potential upper-bounds how valuable that steal would be to the adversary. An important aspect of how we define $\phi$ is that we analyze each of the levels $i$ separately, so that the $i$-stanzas do not have to care about ball moves $(s_k, e_k)$ satisfying $s_k, e_k \ge i + 1$ or satisfying $s_k, e_k \le i - 3$. As we shall see, this decouples the analyses of the levels from one another in several critical ways.

\paragraph{Analyzing the properties of $\phi $.}
At any given moment, let $ A_1 $ denote the set of balls that are present, and, for the sake of analysis, let $ A_2 $ denote a set of $ n $ random balls that are not present, one of which is the ball that will next be inserted. Define $ A = A_1\cup A_2 $. Define $ B = [n]$ to be the set of all non-special slots.

Define a bipartite graph $ G_i = (A, B)$, where for each $ a\in A $ and $ b\in B $ we draw an edge  $(a, b)$ if ball $ a $ is capable of residing in slot $ b $ with level at most $ i $. That is, there is an edge from $a$ to $b$ if $b \in \{h_1(a), \ldots, h_{\Tower(L + i)}(a)\}$. Note that balls $a \in A$ all deterministically have degrees at most $\Tower(L + i)$. For each slot $b$, let $d_i(b)$ denote the degree of $b$ in $G_i$, and call $b$ \defn{high-degree in $G_i$} if $d_i(b) \ge (\Tower (L + i))^c$ for some sufficiently large constant $c$. We call all other nodes in $G_i$ (including all $a \in A$) \defn{low-degree in $G_i$}.

We now argue that most nodes in $G_i$ are far away from any high-degree nodes.
\begin{lemma}\label{lem:safe}
Let $a_1$ be a random ball in $A_1$ and let $a_2$ be a random ball in $A_2$. With probability $1 - 1 / \poly(L)$, neither $a_1$ nor $a_2$ is within distance $O(L)$ of  any high-degree vertex in $G_i$. 
\end{lemma}
\begin{proof}
Since the balls $a \in A$ are independent and randomly selected, the degree $d_i(b)$ is a sum of independent indicator random variables. Moreover, by the near-uniformity of $h$, we know that each $b \in B$ satisfies
\begin{align*}
    \E[d_i(b)] & =  2 \cdot n\Pr_{ x\in U } [h_k (x) = b \text{ for some }k \le \Tower (L + i)] \\
    & = 2 \cdot q(h, \Tower(L + i), b) \\
    & \le 2 \cdot \poly(\Tower(L + i)) \\
    & \le (\Tower (L + i))^c / 2.
\end{align*}
Applying a Chernoff bound, it follows that for all $D \ge (\Tower (L + i))^c$, we have
$$\Pr[d_i(b) \ge D] \le \frac{1}{2^{\Omega(D)}}.$$
Thus
$$\E[d_i(b) \cdot \mathbb{I}_{d_i(b) \ge (\Tower (L + i))^c}] \le \frac{1}{2^{\Omega((\Tower (L + i))^c)}} =  \frac{1}{2^{\poly(\Tower (L + i))}}.$$
This means that the expected sum $S$ of the degrees of the high-degree slots in $G_i$ satisfies
$$\E[S] \le n / 2^{\poly(\Tower (L + i))}.$$
One can also think of $S$ as an upper bound on the number of low-degree nodes in $G_i$ that are adjacent to high-degree nodes in $G_i$. 
Every low-degree node in $G_i$ has degree at most $\poly(\Tower(L + i))$. It follows that the number $\lambda$ of nodes in $G_i$ that are within distance $O(L)$ of a high-degree node satisfies
\begin{align*}
    \E[\lambda] & \le S  \cdot  \poly(\Tower (L + i))^{O(L)}, 
\end{align*}
where the first factor $S$ counts the number of nodes $s$ in $G_i$ that are within distance $1$ of a high-degree node, and the second factor counts the number of $O(L)$-long paths starting at a such a node $s$ and then using only low-degree nodes.
Using our bound on $S$, we get that
$$\E[\lambda] \le  \frac{n}{2^{\poly(\Tower (L + i))}} \cdot  \poly(\Tower (L + i))^{O(L)}.$$
The above quantity is dominated by its first factor, so
$$\E[\lambda] \le \frac{n}{2^{\poly(\Tower (L + i))}}.$$
Applying Markov's inequality, we have that with probability $1 - 1 / 2^{\poly(\Tower (L + i))} \ge 1 - 1 / \poly(L)$,
$$\lambda \le \frac{n}{2^{\poly(\Tower (L + i))}}.$$
Let $a_1$ be a random ball in $A_1$ and $a_2$ be a random ball in $A_2$. The probability that either $a_1$ or $a_2$ is within distance $O(L)$ of a high-degree vertex in $G_i$ is at most
$$\Pr\left[\lambda > \frac{n}{2^{\poly(\Tower (L + i))}}\right]+ \frac{\frac{n}{2^{\poly(\Tower (L + i))}}}{\Theta(n)} = \frac{1}{\poly(L)} + \frac{1}{2^{\poly(\Tower (L + i))}} \le \frac{1}{\poly(L)}.$$
This completes the proof of the lemma.
\end{proof}

The next lemma argues that, if we remove the high-degree nodes from $G_i$, then most of the remaining nodes are far away from any nodes with levels $\ge i + 3$.
\begin{lemma}\label{lem:lowd}
Define $G_i'$ to be the graph $ G_i $, but with all high-degree nodes removed. Let $X$ be the set of balls and non-special empty slots that are currently at a level at least $ i + 3$. For random balls $a_1, a_2$ in $A_1, A_2$, respectively, the probability of either $a_1$ or $a_2$ being within distance $O(L)$ of $X$ in $G_i'$ is at most $1 / \poly(L)$. 
\end{lemma}
\begin{proof}
Note that $X$ is determined by the balls-to-slots scheme, so we will think of $X$ as being selected by an adversary who has full knowledge of $A_1$ and $A_2$ but who has no control over the contents of $A_1$ and $A_2$. Since $\epsilon = 1/n$, the number of slots in $X$ is at most $1$ greater than the number of balls in $X$, so to bound $|X|$, we can focus on the number of balls with levels $\ge i + 3$. 

Each ball $x \in X$ has a level of $i + 3$ or greater, so it has probe complexity at least $\log \Tower(L + i + 2) = \Tower(L + i + 1)$. This means that, at any given moment,
\begin{equation}
    \E[|X|] \le \frac{O(n)}{\Tower(L + i + 1)},
    \label{eq:numhighlevel}
\end{equation}  
where the randomness here comes from the fact that the balls-to-slots scheme guarantees an \emph{expected} average probe complexity of $O(1)$ at any given moment.  By Markov's inequality,
$$|X| \le \frac{n \poly(L)}{\Tower(L + i + 1)}$$
with probability $1 - 1 /\poly(L)$.
The number of nodes $ y $ that are within distance $ O (L) $ of $ X $ in $G_i'$ is therefore at most
\begin{equation}
\frac{n \poly(L)}{\Tower(L + i + 1)} \Tower(L + i)^{O(L)} \ll \frac{n}{\poly(L)}.
\label{eq:numy}
\end{equation}
It follows that, for random balls $a_1, a_2$ in $A_1, A_2$, respectively, the probability of either $a_1$ or $a_2$ being within distance $O(L)$ of $X$ in $G_i'$ is at most $1 / \poly(L)$. 
\end{proof}

We can now argue that each insertion/deletion increases $\phi$ by $O(1)$ (actually $o(1)$) in expectation. Intuitively, this means that insertions/deletions do not, on average, introduce opportunities for the adversary to cheaply achieve a large amount of impact.
\begin{lemma}[Establishing Property 1 for $\phi$]
Each insertion/deletion increases $\phi$ by at most $1 / \poly(L)$ in expectation.
\end{lemma}
\begin{proof}
Consider an insertion of a random ball $ a\in A_2$. Let us consider the effect of the insertion on $\phi_{i + 3}$ for some $i$. Notice that, when $a$ is inserted (i.e., placed into a special slot), $\phi_{i + 3}$ either stays the same or increases by $\le 1$, where the increase comes from the fact that $a$ may be part of some $(i + 3)$-stanza that has positive potential and did not exist before. On the other hand, the only way that $a$ can be part of an $(i + 3)$-stanza that has positive potential is if, in $G_i$, $a$ is within distance $L$ of some slot whose level is $\ge i + 3$---the probability of this occurring is therefore an upperbound on the expected increase to $\phi_{i + 3}$ due to the insertion.

By Lemma \ref{lem:safe}, we have with probability $ 1 - 1 / \poly(L)$ that the set of nodes $y \in G_i$ that are within distance $O(L)$ of $a$ is the same as the set of nodes $y \in G_i'$ that are within distance $O(L)$ of $a$. By Lemma \ref{lem:lowd}, we have that with probability $1 - 1 / \poly(L)$, that within the graph $G_i'$, $ a $ is not within distance $O(L)$ of any node with level $\ge i + 3$ (besides $a$ itself). Thus, with probability $1 - 1 / \poly(L)$, we have that in $G_i$, $a$ is not within distance $O(L)$ of any node with level $\ge i + 3$ (besides $a$ itself). This establishes that, with probability $1 - 1 / \poly(L)$, $a$ is not the starting node for any $(i + 3)$-stanza that has positive potential; and thus the expected increase to $\phi_i$ due to the insertion is $O(1 / \poly(L))$.

Now consider the deletion of a random ball $ a\in A_1 $. By the same reasoning as in the preceding paragraph, with probability $1 - 1 / \poly(L)$, we have that in the graph $G_i$, $a$ is not within distance $O(L)$ of any node with level $\ge i + 3$ (besides possibly $a$ itself). Thus, once $a$ is deleted, we have with probability $1 - 1 / \poly(L)$ that the slot which contained $a$ is not the final slot for any $(i + 3)$-stanza that has positive potential. Hence the expected increase to $\phi_i$ due to the deletion is $1 / \poly(L)$.

In either case, the expected increase to
$$\phi = \sum_{i = 3}^L \phi_i$$
is at most $\sum_{i = 3}^L 1 / \poly(L) = 1 / \poly(L)$. Thus  the lemma is proven.
\end{proof}

Next we analyze the effect that a given move by the adversary has on $\phi $.
\begin{lemma}[Establishing Property 2 for $\phi$]
If a move by the adversary has impact $ r $, it decreases $\phi$ by $r \pm O(1)$. 
\end{lemma}
\begin{proof}
We can assume without loss of generality that the only moves that the adversary ever makes are either (a) to move a ball from a non-special slot into a special slot or (b) to move a ball from a special slot to a non-special slot. Indeed, any move that takes a ball from a non-special slot to another non-special slot can be replaced by a move of type (a) followed by a move of type (b). Also recall that, when a ball is inserted, it initially resides in a special slot, and the adversary can then move it to a non-special slot if desired.

Since moves of type (a) are the reverse of moves of type (b), it suffices to analyze only moves of type (b), and to show that $\phi$ decreases by $r \pm O(1)$.

Suppose the adversary moves ball $ x $ from a special slot $s_1$ to a non-special slot $s_2$, where it has level $j$. Let $ \Sigma $ denote the state of the system before the move, and $ \Sigma' $ denote the state of the system after the move. Let $\phi$ be the potential of $ \Sigma $ and $\phi'$ be the potential of $\Sigma'$.

To complete the proof, we will argue that for each $i \in [L]$:
\begin{itemize}
\item If $i \le j$, then $\phi_i' = \phi_i$.
\item If $i \in \{j + 1, j + 2\}$, then $\phi_i - 2 \le \phi_i' \le \phi_i$.
\item If $i \ge j + 3$, then $\phi_i - 1 \le \phi_i' \le \phi_i - 1 + 1 / L$.
\end{itemize}

\textbf{Case 1: } The first case is immediate, since changes to the positions/levels of balls with levels $\ge i$ do not affect which sequences of ball moves correspond to valid $ i $-stanzas. 

\textbf{Case 2: } Any valid $i$-stanza in $\Sigma'$ is also a valid stanza in $\Sigma$ (hence $\phi_i' \le \phi_i$), but there may be some $i$-stanzas in $\Sigma$ that are not valid in $\Sigma'$ (specifically, any $i$-stanza in $\Sigma$ that makes use of either ball $x$ to start a stanza, or slot $s_2$ to finish a stanza). For any set of disjoint $ i $-stanzas in $ \Sigma $, up to two of those $i$-stanzas might be invalid in $\Sigma'$ (but no more than two!). Thus $\phi_i' \ge \phi_i - 2$. 

\textbf{Case 3: }In the rest of the proof, we focus on the third case, where $i \ge j + 3$. Let $ C $ be a set of disjoint $i$-stanzas in $\Sigma$ that maximizes the sum of the potentials of the $i$-stanzas. Let $s_1 \circ c_1 $ (where $s_1$ is the slot defined earlier in the proof and $c_1$ is a sequence of slots) be the $i$-stanza in $C$ that uses $ x $ as its first ball (if such a stanza exists), and let $c_2 \circ s_2$ (where $c_2$ is a sequence of slots and $s_2$ is the slot defined earlier in the proof) be the $i$-stanza in $C$ that uses slot $s_2$ as its final slot (if such a stanza exists).

We begin by claiming that $c_1$ and $c_2$ exist without loss of generality. If $c_1$ does not exist, then we can modify $ C $ by removing any stanza that uses slot $s_2$, and inserting the stanza $\langle s_1, s_2\rangle$ instead (this replacement either keeps the total potential of $C$ the same or increases it). So $c_1$ exists without loss of generality. If $c_2$ does not exist, then we can modify $C$ by removing any stanza that uses $s_1$, and inserting the stanza $\langle s_1, s_2 \rangle$ instead (again, this cannot decrease the total potential of $ C $). Thus $ c_2 $ also exists without loss of generality. We can further observe that, if $s_1 \circ c_1$ and $c_2 \circ s_2$ happen to be the \emph{same} stanzas as one another, then that stanza is simply $\langle s_1, s_2\rangle$ (indeed, if that stanza were not $\langle s_1, s_2\rangle$, then we could replace it with $\langle s_1, s_2 \rangle$ in order to increase the potential of $C$, which would be a contradiction).

We will now argue that $\phi' \ge \phi - 1$. If $C$ contains the stanza $\langle s_1, s_2 \rangle$, then $C \setminus \{\langle s_1, s_2\rangle\}$ is a set of disjoint $i$-stanzas in $\Sigma'$ with potential exactly $1 - 1 /L$ smaller than that of $C$; thus $\phi' \ge \phi - (1 - 1 / L) \ge \phi - 1$. On the other hand, if $C$ does not contain the stanza $\langle s_1, s_2 \rangle$, then the stanzas $s_1 \circ c_1$ and $c_2 \circ s_2$ must be distinct. In this case, we claim that $c_3 = c_2 \circ s_2 \circ c_1$ is a valid $i$-stanza in $\Sigma'$. Indeed, slot $s_2$ in $\Sigma'$ contains ball $x$ at level $j \le i - 3$, so slot $s_2$ is allowed to be an internal slot in an $i$-stanza; and since, in $\Sigma$, the $i$-stanzas $s_1 \cdot c_1$ (which begins with the slot containing ball $x$) and $c_2 \circ s_2$ (which ends in slot $s_2$) are valid, it follows that, in $\Sigma'$, the $i$-stanza $c_3 = c_2 \circ s_2 \circ c_1$ is valid. Since $c_3$ is a valid $i$-stanza in $\Sigma'$, we have that $C' = C \setminus \{s_1 \circ c_1, c_2 \circ s_2\} \cup \{c_3\}$ is a set of disjoint $ i $-stanzas in $\Sigma'$. The potential of $ C' $ is exactly $ 1 $ smaller than that of $C$. So $\phi' \ge \phi - 1$. 

To complete the proof, we must also establish that $\phi \ge \phi' + 1 - 1/L$. Let $\overline{C}'$ be a set of disjoint $i$-stanzas in $\Sigma'$ that maximizes the sum of the potentials of the $i$-stanzas. If there is no stanza in $\overline{C}'$ that makes use of slot $s_2$, then $\overline{C}' \cup \{\langle s_1, s_2 \rangle\}$ is a valid set of disjoint $i$-stanzas in $\Sigma$, which would mean that $\phi \ge \phi' + 1 - 1/L$. Suppose, on the other hand that there is some $i$-stanza of the form $c_1 \circ s_2 \circ c_2$ in $\overline{C}'$. Then the stanzas $c_1 \circ s_2$ and $s_1 \circ c_2$ are valid in $\Sigma$, and thus $\overline{C}' \setminus \{c_1 \circ s_2 \circ c_2\} \cup \{c_1 \circ s_2, s_1 \circ c_2\}$ is a valid set of disjoint $i$-stanzas in $\Sigma$. This means that $\phi \ge \phi' + 1 \ge \phi' + 1 - 1 /L$, completing the proof.
\end{proof}

The previous two lemmas establish Properties 1 and 2 for $\phi$. Finally, the third property, which states that $0 \le \phi \le Ln$ is trivially true, since $\phi_i \in [0, n]$ for all $i \in [L]$. Thus Theorem \ref{thm:switchinglowerbound} is proven. 

\paragraph{Generalizing to other values of load factor and of probe complexity.}
So far we have assumed for simplicity that $\epsilon = 1/n$ and that the balls-to-slots scheme being analyzed achieves expected average probe complexity $O(1)$. We now generalize our lower bound to consider $\epsilon \ge 1/n$ and probe complexity $\omega(1)$.

\begin{theorem}
Let $L = \lceil (\log^* n) / 2 \rceil $. Consider a universe $U$ of sufficiently large polynomial size. Consider any balls-to-slots scheme that uses nearly uniform probe sequences, that achieves expected average probe complexity $O(\Tower(a))$ (across all balls in the system at any given moment), and that supports some $\epsilon = 1 / \log^{(b)} n$ where $b \le (\log^* n) / 4$.\footnote{In this notation, if $b = 0$, then $\epsilon = 1/n$.} The expected amortized switching cost per insertion/deletion must be at least
$$\Omega(\log^* n - a - b).$$
\label{thm:switchinglowerbound2}
\end{theorem}

Note that, when $a = O(1)$ and $b = 0$, Theorem \ref{thm:switchinglowerbound2} becomes Theorem \ref{thm:switchinglowerbound}, which we have already proven. And as we shall now see, the proof of Theorem \ref{thm:switchinglowerbound2} requires only a slight modification to the proof of Theorem \ref{thm:switchinglowerbound}. 

\begin{proof}
Let us begin by considering $a > 0$ and $b = 0$, so average probe complexity may be $\omega(1)$ but $\epsilon = 1 / n$. 

The only substantive modification to the proof is that, if $a - L \ge 0$, then we redefine any balls in levels less than $a - L$ to now be in level $a - L$ (so we eliminate levels $0, 1, \ldots, a - L - 1$). Intuitively, this is because, since the balls-to-slots scheme is allowed to have average probe complexity $O(\Tower(a))$, it is without loss of generality the case that every ball is in level $a - L$ or above.

Formally, the reason that we need to restrict to levels $i \ge a - L$ is to preserve Lemma \ref{lem:lowd}.
The bound \eqref{eq:numhighlevel} on the expected number of balls with probe complexity at least $\Tower(L + i + 1)$ now becomes
\begin{equation}
    O\left(\frac{n \Tower(a)}{\Tower(L + i + 1)}\right), \label{eq:numhighlevel2}
\end{equation}
instead of $O(n / \Tower(L + i + 1))$. In order for the $\Tower(a)$ term not to become significant in the proof of Lemma \ref{lem:lowd}, we need $\Tower(a) \le \Tower(L + i)$ (that way, in \eqref{eq:numy}, the newly introduced $\Tower(a)$ term can be absorbed into the $\Tower(L + i)^{O(L)}$ term). Since we restrict ourselves to levels $i$ satisfying $i \ge a - L$, Lemma \ref{lem:lowd} continues to be correct for every valid level $i$. 

We must also modify Lemma \ref{lem:impact} to accommodate the fact that each deletion now removes a ball with expected level $\max(0, a - L) + o(1)$ (rather than expected level $O(1)$). This changes our final lower bound on expected average switching cost to $\Omega(L - (a - L)) = \Omega(\log^* n - a)$.

Now suppose we also allow $b > 0$. To handle this, we again modify how we define the levels: we define $\ell_* = L - b - 1$, and we declare any ball or slot (including special slots) that was previously in some level $\ell' > \ell_*$ to now be in level $\ell_*$. The intuition for why we do this is that, once we get to level $\ell_*$, many of the slots that are in that level or above are actually empty slots, so it makes sense to treat that as the top level. 

Formally, the reason that we need to restrict to levels $i \le \ell_*$ is to again preserve \eqref{eq:numhighlevel} in Lemma \ref{lem:lowd}. In particular, \eqref{eq:numhighlevel} must count not just the balls that have probe complexity $\Tower(L + i + 1)$ but also any (non-special) \emph{empty slots} (since such slots represent maximum-level nodes in $G_i'$ and therefore contribute to $|X|$). There may be up to $O(n / \log^{(b)} n)$ such slots (in expectation), each of which is in the top level; to preserve \eqref{eq:numhighlevel}, we therefore need that 
\begin{equation}
    O(n / \log^{(b)} n) \le \frac{O(n \Tower(a))}{\Tower(L + i + 1)}.\label{eq:highlimit}
\end{equation}
Recall, however, that we have limited ourselves to levels $i$ satisfying $i \le l_*$, which implies $i \le (\log^* n) / 2 - b - 1$, and thus that $L + i + 1 \le \log^* n - b$, and therefore that
$$\log^{(b)} n =\Tower(\log^* n - b) \ge \Tower(L + i + 1).$$
Hence, as long as $i$ is a valid level, then \eqref{eq:numhighlevel2} still holds, which preserves the correctness of Lemma \ref{lem:lowd}. 

Since we restrict ourselves to $L - b - 1$ levels, we must also modify Lemma \ref{lem:impact} to accommodate the fact that each insertion now increases the sum of the levels of the balls $J$ by only $L - b - 1$ (instead of by $L$). In the case where $\max(0, a - L) \le L$, this reduces the final lower-bound that we achieve on expected average switching cost to $\Omega(L - b - 1) = \Omega(L)$ (here we are using that $b \le (\log^* n) / 4$), and in the case where $\max(0, a - L) > L$, this reduces the final lower bound to $\Omega(\log^* n - a - b)$. Both lower bounds are equivalent to $\Omega(\log^* n - a - b)$.
\end{proof}

We remark that the restriction $b \le (\log^* n) / 4$ can easily be reduced by defining $L$ to be much smaller than $(\log^* n) / 2$. Such values of $b$ are not relevant to hash-table design, however, since any augmented open-addressing hash table with load factor of at least, say, $1 - 1 / O(\log \log n)$ must use a balls-to-slots scheme that supports $b \le 2$.

\paragraph{Non-nearly-uniform probe sequences. }Finally, we extend our lower bound to non-nearly-uniform probe sequences. To do this, we formally reduce the non-nearly-uniform case to the nearly-uniform case.

For any assignment $A$ mapping some set of up to $n$ balls to slots, and for any function $h$ determining the probe sequences $h_1(x), h_2(x), \ldots$ for each ball, define $ c (A, h) $ to be the total probe complexity needed to implement assignment $ A $ using $h$.
\begin{lemma}
Consider any universe $ U $ and consider any function $h$ assigning a probe sequence to each ball $x \in U$. Then there exists a nearly uniform $ h' $ that has the following guarantee. For any assignment $A$ of $\Theta(n)$ balls to slots, $ c (A, h')\le O (c (A, h) + n) $.
\label{lem:uniform}
\end{lemma}
\begin{proof}

For each ball $ x\in U $ and each $ j\in [n] $, let $ s (x, j) = \argmin_k \{h_k(x) = j\} $ and let $s'(x, j) = \argmin_k \{h'_k(x) = j\}$ (we can assume without loss of generality that these quantities exist). 

We now describe how to construct $ h' $. Rather than specifying $h'_i(x)$ for all $i, x$, it suffices to specify $s'(x, j)$ for all $x, j$. Note that, in order for $s'(x, j)$ to be well defined, the only restriction is that the quantities $s(x, 1), s(x, 2), \ldots, s(x, n)$ must be \emph{distinct} natural numbers.

Define $t:\mathbb{N} \times \mathbb{N} \rightarrow \mathbb{N}$ to be an injective function satisfying $\log t(a, b) \le O(1 + \log a + \log b)$ for all $(a, b) \in \mathbb{N} \times \mathbb{N}$ and satisfying $t(a, b) \ge \max(a, b)$ for all $(a, b) \in \mathbb{N}$. Recall that, for a probe sequence function $h'$, if we have $n$ random balls, then $q(h', i, j)$ is the expected number of balls $x$ that are capable of residing in position $j$ using one of the first $i$ probe values $h'_1(x), \ldots, h'_i(x)$. We set
\begin{equation}
s' (x, j) = t(s (x, j), \lceil q (h, s (x, j), j) \rceil).
\label{sprime}
\end{equation}
We claim that $s'(x, j)$ is well defined. Indeed, if $s'(x, j_1) = s'(x, j_2)$ for some $j_1 \neq j_2$, then we must also have that $s(x, j_1) = s(x, j_2)$, which would be a contradiction. We also observe that $h'$ (constructed using $s'$) has strictly larger probe complexities than does $h$---if a ball $x$ is in a position $j$, then its probe complexity using $h$ would be $\Theta(1 + \log s(x, j))$ but its probe complexity using $h'$ would be $\Theta(1 + \log s(x, j) + \log \lceil q (h, s (x, j), j) \rceil)$. It may seem strange that we are defining $h'$ to be worse than $h$, but as we shall now prove, this allows for us to guarantee that $h'$ is nearly uniform. Once we establish this, then our only remaining task will be to bound how much worse $h'$ is than $h$, in the worst case.

We now argue that $ h' $ is nearly uniform, meaning that $q(h', i, j) \le \poly(i)$ for all $i$. Observe that
\begin{align*}
q(h', i, j) & = n \cdot \Pr_{ x\in U } [h'_k (x) = j \text{ for some }k \le i] \\
            & \le \frac{n \cdot  |\{x \in U \mid s(x, j) \le i \text{ and }  q(h, s(x, j), j) \le i\}|}{|U|},
\end{align*}
since in order to have $h'_k(x) = j$ for some $k \le i$, we must have that $t(s (x, j), \lceil q (h, s (x, j), j) \rceil)) \le i$ and thus that $s(x, j) \le i$ and $q(h, s(x, j), j) \le i$. By expanding out the definition of $q(h, s(x, j), j)$, we get
\begin{align*}
 q(h', i, j)           & \le \frac{n \cdot  |\{x \in U \mid s(x, j) \le i \text{ and }  \Pr_{y\in U } [h_{r}(y) = j \text{ for some }r \le s(x, j)] \le i / n\}|}{|U|}\\ 
            & \le \frac{n \cdot  |\{x \in U \mid s(x, j) \le i \text{ and }  \Pr_{y\in U } [h_{s(x, j)}(y) = j] \le i / n\}|}{|U|}\\ 
            &  \le \frac{n  \cdot |\{x \in U \mid \exists \, s \in [i] \text{ s.t. } h_s(x) = j \text{ and } \Pr_{y\in U } [h_{s}(y) = j] \le i / n\}|}{|U|} \\
            & \le  \frac{n}{|U|} \sum_{s \in [i] \text{ such that}  \Pr_{y\in U } [h_{s}(y) = j] \le i / n} |\{x \in U \mid h_s(x) = j\}| \\
            & = \frac{n}{|U|} \sum_{s \in [i] \text{ such that}  \Pr_{y\in U } [h_{s}(y) = j] \le i / n} |U| \Pr_{y\in U } [h_{s}(y) = j] \\
            & \le \frac{n}{|U|} \sum_{s \in [i] \text{ such that}  \Pr_{y\in U } [h_{s}(y) = j] \le i / n} \frac{i |U|}{n} \\
            &  \le \frac{n}{|U|} \cdot \sum_{s \in [i]} \frac{i |U|}{n} \\
            & = i^2.
\end{align*}
This establishes the near-uniformity of $h'$.

To complete the proof, we must argue that  $ c (A, h')\le O (c (A, h) + n) $. Consider a ball $ x $ that $ A $ assigns to some position $ j $. The probe complexity of $ x $ using $h$ is $1 + \log s(x, j)$, and the probe complexity of $x$ using $h'$ is $O(1 + \log s(x, j) + \log q (h, s (x, j), j))$. Thus, if $ A $ assigns addresses $ j_1, j_2,\ldots, j_m$ to balls $ x_1, x_2,\ldots, x_m$, respectively, for some $m = \Theta(n)$, then our goal is to show that
$$\sum_{i = 1}^m \log q(h, s(x_i, j_i), j_i) \le O\left(n + \sum_{i = 1}^m \log s(x_i, j_i)\right).$$
The cases where $q(h, s(x_i, j_i), j_i) \le \poly(s(x_i, j_i)) $ trivially satisfy $\log q(h, s(x_i, j_i), j_i) \le O(\log s(x_i, j_i))$, so it suffices to show 
$$\sum_{i = 1}^m \mathbb{I}_{q(h, s(x_i, j_i), j_i) > \poly(s(x_i, j_i))} \log q(h, s(x_i, j_i), j_i) \le O(n).$$
Each $j \in [n]$ appears as a $j_i$ at most once in the above sum. Thus 
$$\sum_{i = 1}^m \mathbb{I}_{q(h, s(x_i, j_i), j_i) > \poly(s(x_i, j_i))} \log q(h, s(x_i, j_i), j_i) \le  \sum_{s = 1}^\infty \sum_{j = 1}^n \mathbb{I}_{q(h, s, j) > \poly(s)} \log q(h, s, j).$$
We can therefore complete the proof by showing that
$$\sum_{s = 1}^\infty \sum_{j = 1}^n \mathbb{I}_{q(h, s, j) > \poly(s)} \log q(h, s, j) \le O(n).$$
Let $Q_s$ be the set of $j$ for which $q(h, s, j) > \poly(s)$. Then,
$$\sum_{s = 1}^\infty \sum_{j = 1}^n \mathbb{I}_{q(h, s, j) > \poly(s)} \log q(h, s, j) = \sum_{s = 1}^\infty \sum_{j \in Q_s} \log q(h, s, j).$$
For any fixed $s$, we have that
\begin{align*}
\sum_{j = 1}^n q(h, s, j) & = \sum_{i = 1}^s \sum_{j = 1}^n n\Pr_{x \in U}[h_i(x) = j] \\
                          & = n \sum_{i = 1}^s \sum_{j = 1}^n \Pr_{x \in U}[h_i(x) = j] \\
                          & = n \sum_{i = 1}^s \Pr_{x \in U} [h_i(x) = j \text{ for some } j \in [n]] \\
                          & \le n s. 
\end{align*}
Thus $\sum_{j \in Q_s} q(h, s, j)$ is also at most $sn$ and has at most $sn / \poly(s) = n / \poly(s)$ terms. 
By Jensen's inequality, this implies that 
$$\sum_{j \in Q_s} \log q(h, s, j) \le \frac{n}{\poly(s)} \log \frac{sn}{n / \poly(s)} = \frac{n}{\poly(s)}.$$
Summing over all $s$,
$$\sum_{s = 1}^\infty \sum_{j \in Q_s} \log q(h, s, j) \le \sum_{s = 1}^\infty \frac{n}{\poly(s)} = O(n).$$
This completes the proof.
\end{proof}

By the preceding lemma, the assumption in Theorem \ref{thm:switchinglowerbound2} that $h$ is nearly uniform is true without loss of generality, since we can substitute any non-nearly-uniform $ h $ with a nearly-uniform $ h' $ while having an asymptotically negligible effect on the probe complexity of any balls-to-slots assignment. Thus we arrive at the main theorem of the section:

\begin{theorem}
Let $L = \lceil (\log^* n) / 2 \rceil $. Suppose the universe $U$ has sufficiently large polynomial size. Consider any balls-to-slots scheme that achieves expected average probe complexity $O(\Tower(a))$ (across all balls in the system at any given moment) and supports some $\epsilon = 1 / \log^{b} n$ where $b \le (\log^* n) / 4$. The expected amortized switching cost per insertion/deletion must be at least
$$\Omega(\log^* n - a - b).$$
\label{thm:switchinglowerbound3}
\end{theorem}

\begin{corollary}
 Suppose the universe $U$ has sufficiently large polynomial size. Consider any balls-to-slots scheme that achieves expected average probe complexity $O(1)$ (across all balls in the system at any given moment) and supports $\epsilon = 1/n$. The expected amortized switching cost per insertion/deletion must be $\Omega(\log^* n)$.
\end{corollary}

To conclude the section, we reinterpret our result as a lower bound on augmented open-addressing.

\begin{corollary}
Any augmented open-addressed hash table that stores quotiented $(1 + \Theta(1)) \log n$-bit elements in an array and incurs $O(\log^{(k)} n)$ expected wasted bits per key must have average insertion/deletion time $\Omega(k)$. 
\label{cor:augmentlower}
\end{corollary}
\begin{proof}
We can assume without loss of generality that $k \ge 2$. In order for the wasted bits per key to have expected value $O(\log^{(k)} n) \le O(\log \log n)$, the load factor $1 - \epsilon$ of the array must satisfy $1 -\epsilon \ge 1 - \frac{\log \log n}{\log n}$. That is, the balls-to-slots scheme used by the hash table must support $\epsilon \le \frac{\log \log n}{\log n}$. In the language of Theorem \ref{thm:switchinglowerbound3}, this means that $b < 2$. 

The bound of  $O(\log^{(k)} n)$ wasted bits per key also implies that the (expected) average probe complexity of the balls-to-slots scheme is $O(\log^{(k)} n)$. In the language of Theorem \ref{thm:switchinglowerbound3}, this means that $a \ge (\log^* n) - k$. 

Applying Theorem \ref{thm:switchinglowerbound3}, we get that the average switching cost of the balls-to-slots scheme is at least $\Omega(\log^* n - a - b) = \Omega(k)$. Thus the average insertion/deletion time of the hash table is $\Omega(k)$. 
\end{proof}

\section{Encoding Metadata in an Augmented Open-Addressed Hash Table}\label{sec:metadata}

So far, we have computed  tight bounds on the probe complexity of any balls-to-slots scheme. If the balls-to-slots scheme used by an augmented open-addressing hash table has total probe complexity $\ell$, then the hash table must store at least $\Omega(\ell)$ bits of metadata. In this section, we present general machinery for how to implement the metadata of the hash table to use exactly $O(\ell)$ bits, while also allowing for constant-time modifications to the metadata. The key difficulties here are that $\ell$ may differ for between elements (i.e., it is nonuniform) and that $\ell$ may be, on average, very small, meaning that we cannot afford a high space overhead per element.

To address these issues, we develop two fundamental building blocks: the first is a data structure that we call the \defn{mini-array}, which compactly stores a $\polylog n$-element dynamic array of items that are between $ 1 $ and $ O (\log n) $ bits each so that array entries can be queried and modified in constant time; the second is a data structure that we call the \defn{local query router}, which compactly stores routing information (i.e., information about where some element $ x $ can be found in the hash table) for up to $ O ((\log n)/\log\log n) $ elements at a time, while supporting queries/updates to the routing information in constant time. 

As foreshadowing, and to give some additional context, let us comment on how these building blocks will be used later. Ultimately, our approach to storing metadata in a hash table will be the following: we will hash keys to buckets of some expected size $K = \polylog n$; then, within each bucket, we will hash keys to $K$ smaller buckets of expected size $O(1)$; for each of these smaller buckets, we will use a local query router to store the metadata for the elements in that bucket; and for each of the larger buckets, we will use a mini-array to store the $ K $ local query routers for its $ K $ smaller buckets. In this section, however, our goal is simply to construct mini-arrays and local query routers.

\subsection{Preliminaries: Implementing machine-word operations with lookup tables}
Several of the data structures in this section will make use of the \defn{lookup-table technique} (sometimes also called the Method of Four Russians). This allows for us to implement potentially complicated operations on $(\log n)/2$-bit inputs in constant time.

More formally, call such a function $f(x_1,\ldots, x_j)$ \defn{lookup-table-compatible} if:  the input tuple $(x_1, \ldots, x_j)$ takes less than $(\log n) / 2$ bits; the output takes $O(\log n)$ bits; and $f$ can be evaluated in time $O(n^{1/4})$. 

If $ f $ is lookup-table-compatible, then, when we initialize a hash table of size $n$, we can pre-construct a lookup table $L$ of size $\sqrt{n}$ such that $L[x_1, \ldots, x_j] = f(x_1, \ldots, x_j)$ for each of the up to $\sqrt{n}$ input tuples $(x_1, \ldots, x_j)$. The lookup table allows for us to evaluate $f$ in constant time during hash-table operations. The lookup table $L$ consumes at most $\tilde{O}(\sqrt{n})$ bits of space and can be constructed in time at most $O(n^{3/4})$. 

We can also rebuild the lookup table (in a deamortized fashion) whenever the parameter $n$ changes by more than a constant factor, so the restriction that each input tuple $(x_1, \ldots, x_j)$ takes less than $(\log n) / 2$ bits is always a function of the current $n$. 

Finally, suppose that we have a function $f$ for which the input tuple $(x_1, \ldots, x_j)$ takes $\Theta(\log n) $ bits, rather than $ (\log n)/2$ bits. We say that $f$ is \defn{lookup-table-friendly} if for some positive constant $ c $, there exist lookup-table-compatible functions $f_1, \ldots, f_c$ such that: the input tuple $(x_1, \ldots, x_j)$ can be decomposed into $(\log n) / 2$-bit (or smaller) pieces $P_1, \ldots, P_c$, and $f(x_1, \ldots, x_j)$ can be computed in constant time given $f_1(P_1), \ldots, f_c(P_c)$. To implement $f$ in constant time, we can implement each $f_i$ using the lookup-table technique. So lookup-table-friendly functions can be evaluated in constant time without loss of generality.

\subsection{Storing a mini-array of variable-size values}
Consider the following basic data-structural problem, which we call the \defn{mini-array problem}. Let $c$ be a sufficiently large positive constant, and let $K = \log^c n$. We wish to store a $ K $-element \defn{mini-array} $A[1], \ldots, A[K]$, where each element $A[i]$ has some size $s_i \in [0, O(\log n)]$ bits. We wish to support queries (i.e., tell me $A[i]$) and updates (i.e., set $A[i]$ to a new value) in constant time, and we wish to use space $O(K + \sum_i s_i)$ bits. In our setting, we will have a large collection of mini-arrays, each a part of a larger data structure whose total size is $\Omega(n)$.  So we will allow for our solution to use lookup tables that are shared among all of the mini-arrays.

How should we implement a mini-array? The problem is that the sizes $s_i$ of the elements in the array are non-uniform and change over time. So we cannot implement $ A $ as a standard array. Instead, we take inspiration from the external-memory model \cite{vitter2001external}, and we implement $ A $ as a B-tree \cite{bayer2002organization} $T$. The basic idea is that, since we can implement (most) operations on $\Theta(\log n)$-bit machine words in constant time using the lookup-table approach, we can think of machine words as representing data blocks in the external-memory model.

The tree $T$ consists of $\polylog n$ nodes, each of which is $\Theta(\log n)$ bits (the only exception is the root node, which may contain fewer bits). Because the tree consists of only $\polylog n$ nodes, pointers within the tree need only be $\Theta(\log \log n)$ bits each.\footnote{For the applications in this paper, the amount of memory needed to implement $T$ will always be known (up to constant factors) up front, so we can preallocate the memory in a single contiguous array. Even if the size of $T$ is not known up front, however, it is still possible to implement pointers within the tree using $\Theta(\log \log n)$ bits per pointer. Indeed, we can assign the nodes distinct $\Theta(\log \log n)$-bit identifiers, and then we can maintain a dynamic fusion tree \cite{patrascu2014dynamic} mapping identifiers to true $\Theta(\log n)$-bit pointers---the fusion tree allows us to perform address translation in order to go from an identifier to the corresponding actual node. Note that the fusion tree introduces only a constant-factor space overhead overall, and introduces only on additive constant time overhead for each operation; so we can feel free to ignore the fusion tree, and treat pointers as each using $\Theta(\log \log n)$ bits.} 

Each internal node of $ T $ stores $\Theta(\log n / \log \log n)$ pointers to children (although, again, the root may contain fewer), and for each child the node stores two pivots $p_1, p_2 \in [K]$ indicating the range of indices that the child covers. Each leaf of $T$ stores $\Theta(\log n)$ bits of array entries (i.e., $A[i], \ldots, A[j]$ for some $i, j$ such that $\sum_{\ell = i}^j (s_\ell + 1) = \Theta(\log n)$). We call these the \defn{array bits}. Each leaf also stores a $\Theta (\log n) $-bit bitmap indicating where each $A[\ell]$ begins within the array bits.

Nodes are merged and split as in a standard B-tree: there is some positive constant $ d $ such that, whenever a node exceeds $d \log n$ bits, the node is split into two nodes, and whenever a node's size falls below $d (\log n) / 2$ bits, the node is merged with one of its neighbors (and then the new merged node may also need to be split). The only way that the height of the tree can increase is if the root splits into two nodes $a$ and $b$ (in which case a new root is created with $a$ and $b$ as children), and the only way that the height of the tree can decrease is if the root has only a single child, in which case the root is eliminated. Every node except the root has the property that it  always uses $\Theta(\log n)$ bits, but the root may be smaller (since it has no neighbors that it can merge with). Since the tree has fanout $\Theta((\log n) / \log \log n)$ (for all internal nodes except for the root), and since the tree consists of $O(K) \le \polylog n$ nodes, the depth is $O(1)$.

Using the lookup-table approach, we can implement both queries and updates in constant time. In particular, the tasks of navigating down the tree, finding where a given $A[i]$ resides in some leaf, modifying some $A_i$ in some leaf, and modifying internal nodes are all directly implementable using lookup-table-friendly functions. 

This concludes the description of how to implement a mini-array. Each operation is deterministically constant time. And, up to constant factors, the space-usage of the tree is dominated by the leaves, which in aggregate use $O(K + \sum_i s_i)$ space, as desired. The lookup tables used to implement the mini-array take a total of $\tilde{O}(\sqrt{n})$ space, but since these lookup tables can be shared across all instances of mini-arrays, that space is negligible. 

\subsection{Storing routing metadata for collections of $O((\log n) / \log \log n)$ elements}

We now describe a second data-structural problem, which we call the \defn{query-router problem}. To understand the query-router problem, it is helpful to understand how we will use local query routers in our hash tables. We will hash $\Theta(n)$ keys to $\Theta(n)$ different local query routers, and each local query router will be responsible for storing the probe-indices corresponding to those keys---that is, if a local query router stores a key $x$ that resides in slots $h_i(x)$ of the hash table, a query searching for key $x$ must be able to recover the value $i$ from the local query router. The way in which the local query router is used results in several interesting properties that we will make exploit in its construction: with high probability in $n$, each local query router will be storing information for at most $O(\log n / \log\log n)$ keys; additionally, if a local query router wishes to access one of the keys $x$ for which it is storing information, it can do so in constant time (without having to actually store $x$). 

With these properties in mind, we now formally define the query-routing problem: Consider a set $S$ of distinct keys, and a function $f:S \rightarrow \mathbb{N}$ that maps keys to distinct values. We wish to support modifications to $S$ and $f$ (i.e., delete $s$ from $S$, or insert $s$ into $S$ with $f(s) := u$) and $f$-evaluation queries (i.e., what is $f(s)$ for some specific $s \in S$?) in constant time (with high probability in $ n $). We are guaranteed that $|S|$ never exceeds $O((\log n) / \log \log n)$ and that $f(s)$  always takes $O(\log \log n)$ bits. Setting $r = |S| + \sum_{s \in S} \log f(s)$ to be the sum of the sizes of the $f(s)$'s, we wish to have a data structure of \emph{expected size} $O(r)$ bits, at any given moment, and of \emph{worst-case size} $O(\log n)$ bits, at any given moment with high probability in $ n $. Our data structure also has access to a constant-time oracle for the function $g = f^{-1}$. That is, if $f(s) = u$ for some $s$, then the oracle function satisfies $g(u) = s$. (If $f(s) \neq u$ for all $s \in S$, then $g(u)$ is not defined, and could return an arbitrary value.) The oracle makes it so that our data structure does not have to store keys---it can recover each key based on the corresponding $f$-value.

We now describe a data structure, which we call a \defn{local query router}, that solves the above problem. Although the precise specifications are slightly different, the design for the local query router is very similar to the querying mechanism used in past work on adaptive filters \cite{bender2018bloom} (as well as by other subsequent work on succinct filters \cite{liu2020succinct}). 

Before we continue, let us make some simplifications to the requirements of a local query router, and argue that these simplifications are without loss of generality. First, we may assume that the local query router has a lifespan of only $O(\log n / \log \log n)$ operations, since we can rebuild the local query router from scratch once every $O(\log n / \log \log n)$ operations (and in a deamortized fashion). Second, it suffices to construct a local query router with failure probability $1 - 1 / n^\epsilon$ on any given insertion/deletion, since we can amplify this failure probability to $1 / \poly(n)$ by storing $O(1)$ independent local query routers, and keeping track of which one(s) haven't yet failed---in any sequence of $O(\log n / \log \log n)$ operations, the probability of all $O(1)$ local query routers failing is $1 - 1 / \poly(n)$. We call a local query router that makes the above simplifications a \defn{simplified local query router}.

To construct a simplified local query router, we will need the following basic lemma about binary tries.
\begin{lemma}
Let $k = O(\log n / \log \log n)$ and let $r_1, \ldots, r_k$ be random binary strings. Let $T$ be the binary trie storing the smallest unique prefix of each $r_i$ (i.e., if the smallest unique prefix of $r_1$ is $01101$, then there is a path corresponding to $01101$ in the trie). Then $ T $ has expected size $ O (k) $, and for any constant $c > 1$ there exists a constant $ \epsilon > 0$ such that with probability $1 - 1/n^\epsilon$, $T$ has size that most $\log n / c$ bits. 
\label{lem:binarytrie}
\end{lemma}
\begin{proof}
Imagine constructing $ T $ by inserting each of $r_1, \ldots, r_k$ into the trie one after another. Inserting a new element into $T$ corresponds to performing a random walk down the tree $T$ to some leaf $\ell$, and then appending a path of some length $X$ below that leaf, and then adding two new leaves at the end of that path. Note that the random variable $X$ is independent between insertions and satisfies
$$\Pr[X \ge i] = 1/2^i.$$
Thus, once all of the $k$ insertions are performed, the size of $ T $ is simply a sum of independent geometric random variables. By a Chernoff bound for sums of independent geometric random variables, the lemma follows.
\end{proof}

We can now construct a simplified local query router. We hash of the elements of $S$ to random binary strings, and we place those binary strings in a trie $T$. For each leaf of $T$ corresponding to some $s \in S$, we also store the value $f(s)$ at that leaf. 

In more detail, we can encode the tree, along with the $f$-value for each of the leaves as follows. Perform a depth-first traversal through the tree, and write down the sequence of moves that the traversal performs (i.e., moves of the form "go to left child", "go to right child", "go up"); call this portion of the encoding $E_1$, and observe that $|E_1|$ is $\Theta(|T|)$ bits. Then write down the $f$-values for the leaves in the same order that they appear in the depth-first traversal of the tree (it is straightforward to encode the value in such a way that it can easily be determined where one value begins and another finishes); call this portion of the encoding $E_2$, and observe that $|E_2|$ is $\Theta(|S| + \sum_{s \in S} \log f(s))$ bits. 

The total number of bits in the encoding is $|E_1| + |E_2| = O(|T| +  |S| \sum_{s \in S} \log f(s))$ which, by Lemma \ref{lem:binarytrie}, has expected value $O(|S| +  \sum_{s \in S} \log f(s))$. Lemma \ref{lem:binarytrie} further tells us that, for any positive constant $c$, there exists a positive constant $\epsilon$ such that $|E_1| \le (\log n) / c$ with probability $1- 1/n^{\epsilon}$. Since, by assumption, we have that $|E_2| = O(\log n)$, it follows that the total encoding takes $O(\log n)$ bits. Finally, since $|E_1| \le (\log n) / c$, and since $E_2$ can be broken into $O(1)$ lists of $f$-values that are $(\log n) / c$ bits each, we can implement insertions/deletions/queries on the encoding in constant time using lookup-table-friendly functions. Thus we have constructed a constant-time simplified local query router, and since the reduction from a full local query router to a simplified local query router is without loss of generality, we have also completed the construction and analysis for the full local query router.

\section{An Optimal Augmented Open-Addressed Hash Table}

Using the techniques developed in the previous sections, we can now construct a dynamically-resized augmented open-addressed hash table that stores $\Theta(\log n)$-bit key-value pairs, that supports insertions/deletions in time $O(k)$, that supports queries in time $O(1)$, and that achieves $O(\log^{(k)} n)$ wasted bits per key. (The running-time and space guarantees are with high probability in $n$). By Corollary \ref{cor:augmentlower}, our data structure achieves the best possible tradeoff curve between time and space that any augmented open-addressed hash table can achieve.

We break the section into three parts:
\begin{itemize}
    \item Subsection \ref{sec:fixed} constructs a fixed-capacity hash table that uses $nw + O(n\log^{(k)} n)$ bits of space to store $n$ $w$-bit keys-value pairs. 
    \item Subsection \ref{sec:variable} shows how to make the hash table dynamically-resizable.
    \item And Subsection \ref{sec:quotient} reduces the space consumption to be within $O(n\log^{(k)} n)$ bits of the information-theoretic optimum.
\end{itemize}

\subsection{Turning the $k$-\kicktree into a hash table}\label{sec:fixed}

In this section, we construct a fixed-capacity hash table that uses $nw + O(n\log^{(k)} n)$ bits of space to store $n$ $w$-bit keys-value pairs. 

\paragraph{The layout.}
Let $K = \polylog n$ be a parameter.  We hash keys to $(1 + 1 / K^{1/3}) (n/K) $ bins, each of which we refer to as a \defn{\cubby}. With high probability in $ n $, each \cubby receives at most $K$ keys at any particular time. 

Each \cubby maintains a \defn{storage array} capable of storing up to $K$ keys/values. Keys are assigned a position in the storage array using the $k$-\kicktree from Theorem \ref{thm:probeupper} for some parameter $k$. (We will discuss how to do this time-efficiently later.)  The parameter $ k $ will determine the tradeoff between time and space efficiency in our data structure. 

Recall that the $k$-\kicktree associates each key $ x $ with a random sequence of hash functions $g_0(x), \ldots, g_k(x)$, where each $g_{i + 1}(x)$ is a child bin of $g_i(x)$. Of course, $g_k(x)$ determines all of $g_1, \ldots, g_{k - 1}(x)$, and one way to pick $g_k(x)$ is to select a random $g(x) \in [K]$, and set $g_k(x)$ to be the  depth-$k$ bin that contains position $g(x)$. We will refer to $ g (x) $ as $ x $'s \defn{preferred slot} (within the \cubby). 

For each \cubby, and for each $ i\in [K] $, we maintain a local query router that stores metadata for the keys who have preferred slot $g(x) = i$. For each such key $x$, the local query router stores the index $j$ such that $x$ is in position $h_j(x)$ of the \cubby---if $x$ is stored at depth $i$ by the $k$-\kicktree, then we can store $j$ using $O(\log^{(i + 1)} K) = O(\log^{(i + 2)} n)$ bits. As a slight abuse of notation, to simplify discussion throughout the rest of the paper, we shall redefine the probe complexity of $x$ to be exactly $\Theta(\log^{(i + 1)} K)$, even though technically the true probe complexity may be smaller. 

We store the $K$ local query routers in a mini-array $A$. The result is that any key $ x $ in the data structure can be recovered by (a) hashing to the appropriate \cubby; (b) finding the $g(x)$-th local query router in the mini-array; and (c) using that local query router to determine which slot of the storage array the key resides in. Note that the array $A$ is local to each individual \cubby.

\paragraph{Implementing insertions/deletions/queries in constant time.}
We have already seen how to implement queries in constant time using the mini-array $ A $ of local query routers. Deletions can also be implemented in constant time by simply removing the key/value pair.

Insertions are slightly more tricky, however. Recall that the balls-to-slots scheme has $ k +1 $ classes of bin sizes, where the sizes are denoted $ s_0,\ldots, s_k$. Note that, in this setting, $ s_0 = K = \polylog n$ and $s_i = \poly(\log^{(i + 1)} n)$ for each $i \in [k]$. 

Let us start by ignoring depth $ 0 $ and discuss how to implement depths $ 1,\ldots, k $. We maintain a second mini-array $M$ storing metadata for each of the $ K/s_1 $ depth-1 bins. For each such bin, the metadata that we store is the information of which slots are free in that bin, and for each slot that is not free in that bin, what the depth is for the element in that slot. In aggregate, this information comprises $\poly (\log\log n) $ bits. Using this metadata, along with the mini-array $ A $, we can use lookup-table-friendly functions to implement the portions of an insertion that occur in depths $ 1,\ldots, k $ in time $ O (k) $ (i.e., we can perform the entire insertion, except possibly the final step in which we must find a free slot to place some depth-0 element in). 

The only task that remains is to locate a free slot in depth-0 (i.e., in the entire \cubby). For this, we can simply maintain a $\log_{\log n} K = O(1)$-depth tree with uniform fanout $\log n $, in which each internal node stores a $\log n $-bit bitmap indicating which of its children contain at least one free slot, and each leaf stores a $\log n $-bit bitmap indicating which of the $\log n $ slots corresponding to that leaf are free. We refer to this as the \defn{free-slot tree}. The free-slot tree supports constant-time modifications and queries (where a query finds a free slot). 

\paragraph{Proving correctness. }
We now establish the correctness of our data structure. 

\begin{lemma}
The above data structure correctly implements insertions/deletions/queries, ensures that insertions/deletions take time $O(k)$ with high probability in $n$, and ensures that queries take time $O(1)$ deterministically.
\end{lemma}
\begin{proof}
By a Chernoff bound, each \cubby has at most $ K $ keys at any specific time, so each insertion has a high probability of hashing to a \cubby that has room for it. This means that the $k$-\kicktree can operate correctly without overflowing.

We next verify that each of the local query routers operates correctly: each local query router requires that it store metadata for at most $ O (\log/\log\log n) $ keys, and that each key has at most $ O (\log\log n) $ bits of metadata. The first requirement follows by a Chernoff bound on the number of keys that hash to a given \cubby and have a given value of $ g (x) $. (The number of such keys has expected value $ 1 $, and is at most $ O (\log n/\log\log n) $ with high probability in $ n $.) The second requirement follows from the fact that each key has probe complexity $ O (\log K) = O (\log\log n) $ bits in the balls-to-slots scheme.

Next we verify that each mini-array operates correctly: each mini-array requires that its entries are each $ O (\log n) $ bits. This is immediate for $ M $, and for $ A $ it follows from the fact that each local query router takes $ O (\log n) $ bits (with high probability).

Since the requirements for correctness have been met for each mini-array and query-router, all of them will support constant-time operations with high probability in $n$. It follows that insertions and deletions are correct and take $O(k)$ time (with high probability)\footnote{In the low-probability event that an insertion fails to be implementable, either because a cubby overflows, or because a query-router overflows, we simply rebuild the entire data structure from scratch.}, and that queries are correct and take $O(1)$ time deterministically.
 \end{proof}

Finally, we analyze the space consumed by our data structure. We shall assume that the number $w$ of bits taken by each key/value pair satisfies $w = \Theta(\log n)$.

\begin{lemma}
With high probability in $n$, the size of the data structure is $nw + O(n \log^{(k)} n)$ bits.
\label{lem:space}
\end{lemma}
\begin{proof}
We start by bounding the space consumed by storage arrays. There are $ (1+1/K ^ { 1/3 }) n/K $ \cubbys each of which has a storage array of size $ Kw$ bits. This reduces to 
$$(1+1/K ^ { 1/3 }) nw = nw + o(n)$$
bits.

Next we bound the expected space consumed by the mini-arrays $A$ and $M$ in a \cubby. If a \cubby has $j \le K$ keys and their probe complexities are $ a_1, \ldots, a_{j}$, then the expected space consumed by $A$ and $M$ is $ O (K + \sum_{i } a_i)$ bits (this is expected rather than worst-case because the local query routers may add more bits in the worst case). On the other hand, by Theorem \ref{thm:probeupper}, 
$$\E\left[K +\sum_{i } a_i\right] = O(K \log^{(k + 2)} n) = O(K \log^{(k)} n).$$
Thus the expected amount of space used by the mini-arrays in a given \cubby is $O(K \log^{(k)} n) $. Summing over the \cubbys, the total amount of space used by mini-arrays in the data structure is 
$$O(n \log^{(k)} n)$$
bits in expectation.

We can turn this into a high-probability bound as follows. Define $r_1, \ldots, r_{(1 + 1/K^{1/3})n/K}$ so that $r_i$ is the number of bits consumed by the mini-arrays in the $i$-th \cubby. Note that each $r_i$ is deterministically at most $O(K \log n)$, since the data structure is rebuilt whenever either (1) more than $K$ elements simultaneously hash to some \cubby, or (2) $\omega(\log n)$ bits are needed for some local query router in some \cubby. Moreover, regardless of the outcomes of $\{r_j \mid j \neq i\}$, we have that $\E[r_i] = O(K \log^{(k)} n)$. Thus we can apply a Chernoff bound to deduce that $\sum_i r_i$ is tightly concentrated around its mean, so the total space used by mini-arrays is $O(n \log^{(k)} n)$ bits with high probability in $n$. \end{proof}

Putting the pieces together, we have the following theorem:
\begin{theorem}
Let $w = \Theta(\log n)$ and $k\in [\log^* n]$. One can construct a dictionary that stores up to $n$ $w$-bit key/value pairs, while supporting insertions/deletions in time $ O (k) $, supporting queries in time $O(1)$, and using total space $wn + O(n \log^{(k)} n)$ bits, with high probability in $n$. 
\end{theorem}

The preceding theorem has several limitations that we will remove in the coming sections. The first limitation is that our hash table does not yet support dynamic resizing (i.e., it has a fixed capacity). The second limitation is that our hash table stores each key/value pair in its entirety, even though information-theoretically, only $ w -\log n + O (1) $ bits are needed per key/value pair. Each of the next two sections will remove one of these constraints.

We conclude the section by proving a simple technical lemma about \cubbys that will be useful later. The lemma says that, even though modifying a \cubby takes time $ O (k) $, we can build a \cubby from scratch in linear time $ O (K) $.
\begin{lemma}
Let $ S $ be a set of at most $ K $ key/value pairs. We can construct a \cubby storing $ S $ in time $ O (K) $ with high probability in $ n $
\label{lem:buildicecube}
\end{lemma}
\begin{proof}
Recall that keys are stored in one of $ k +1 $ depths. Inserting a key into depth $ i $ takes up to $ O (i) $ time, since we may have to relocate one key in each of depths $ i -1,\ldots, 0 $. To get around this issue, we build the \cubby as follows: we first try to place each key into depth $ k $, and if a key cannot be placed in depth $ k $ (either because there is no room, or because the key has hash $s(x) < k$, then we do not insert the key); we then try to place the remaining keys into depth $ k -1 $, and again if a key cannot be placed into depth $ k -1 $, then we do not insert the key; we continue like this for each of depths $ k -2, k -3,\ldots, 0 $ one after another.

For each key $ x $, define $ j_x $ so that $ k - j_x +1 $ is the depth at which $ x $ ends up being inserted. The total time to build the \cubby is $ O (\sum_{x }^{} j_x) $. Define $ r_x $ to be the probe complexity of $ x $. Then $j_x \le O(1) + r_x / \log^{(k)} n$. Thus the total time to build the \cubby is
$$O(K) + O(\sum_x r_x / \log^{(k)} n).$$
We know from the analysis in Theorem \ref{thm:probeupper} that $\E[\sum_x r_x] \le O(K \log^{(k)} n)$, so the expected time to build the \cubby is $ O (K) $.

To turn this into a high-probability bound, we must obtain a high-probability bound on  $\sum_x j_x$. For this, we can perform a similar analysis as in Theorem \ref{thm:probeupper2}. Break the \cubby into $\sqrt{K}$ parts. Since there are at most $K = \polylog n$ elements total, the number of elements that hash to any given part is at most $\sqrt{K} (1 + 1 / \polylog n)$ with high probability in $n$. If a part receives more than $\sqrt{K}$ keys, then call the remaining $K / \polylog n$ keys that it receives \defn{extra keys}. Modify the construction of the \cubby so that we first find places for all of the non-extra keys, and then we insert the extra keys---since there are so few extra keys, they add a negligible total amount to the running time. Define $J_1, \ldots, J_{\sqrt{K}}$ so that $J_i$ is the sum of the depths of the up to $K$ non-extra keys that map to the $i$-th part. By the same analysis as above, we have that $\E[J_i] = O(\sqrt{K})$ for each $i$, regardless of the outcomes of the outcomes of $\{J_r \mid r \neq i\}$. We also have that $J_i \le O(k \sqrt{K})$ deterministically. Thus we can apply a Chernoff bound to $J = \sum_{i = 1}^{\sqrt{K}} J_i$ to determine that $J = O(K)$ with high probability in $n$. This implies that the total construction time for the \cubby is $O(K)$, as desired. 
\end{proof}

\subsection{Supporting dynamic resizing}\label{sec:variable}

In this section, we adapt the hash table from the previous section in order to support dynamic resizing: the amount of space that the hash table consumes will now be a function of the current number $ n $ of elements in the table, rather than some maximum capacity $ n $.

To begin, we will focus on supporting $n$ in a fixed range $[N, 2N]$, and we shall assume that the size of a key-value pair is $w = \Theta(\log N)$ bits. At the end of the section, we will generalize to allow for $n$ to vary over a polynomial range (i.e., it is subject only to the constraint that $\log n = \Theta(w)$). (And, in fact, later in Section \ref{sec:largekeys}, we will show how fully generalize for arbitrary values of $n$.)

\paragraph{The basic layout.}
Let $K = \polylog N$ and let $k \in [\log^* K]$ be a parameter. Our hash table will consist of $ N/K $ \defn{\facilitys}, where each \facility contains $\Theta(K)$ elements. When an element is inserted, it is hashed to a random \facility.

Each \facility is composed of many \cubbys (implemented as in the previous section) of different sizes. More specifically, at any given moment, we will always maintain a \defn{distribution invariant}, which guarantees that for each \facility there are:
\begin{itemize}
\item $\Theta((\log^{(k)} n)^2)$ \cubbys of capacity $K / (\log^{(k)} n)^2$;
\item and $\Theta\left(\frac{(\log^{(j)} n)^2}{(\log^{(j + 1)} n)^2}\right)$ \cubbys of size  $K / (\log^{(j)} n)^2$, for each $j \in [k-1]$.
\end{itemize}
We say that an \cubby is \defn{$j$-tiered} if its size is $K / (\log^{(j)} n)^2$. The way to think about the distribution of \cubby sizes is that, for each $j < k$, the \emph{total size} of the $ j $-tiered \cubbys is asymptotically equal to the size of a \emph{single} $ j +1 $-tiered \cubby. That is, for $j < k$, there are at most $O(K / (\log^{(j + 1)} n)^2)$ elements in $ j $-tiered \cubbys at a time.

At any given moment, one of the $ 1 $-tiered \cubbys is designated as the \defn{tail}. The second invariant that we will maintain is that, at any given moment, \emph{all} of the \cubbys except for the tail are completely full. We call this the \defn{saturation invariant}.

We will describe how to efficiently maintain the distribution and saturation invariants shortly, but first we finish describing the layout of a \facility. Each \facility must always store the following: (a) pointers to all of the \cubbys stored in the \facility; and (b) metadata allowing for queries to determine which \cubby the key they are looking for is in. Since each \cubby has size $\polylog n $, the pointers to \cubbys take negligible space. The metadata for queries can be stored as follows: we maintain a mini-array $ D $ with $ K $ entries; we hash each key $ x $ to one of the entries of the mini-array, and each entry stores a local query router that maps each key $ x $ to the appropriate \cubby. Note that, if a key $x$ is in a $k$-tiered \cubby, then we can can indicate which $k$-tiered \cubby it is in using $$O(\log (\log^{(k)} n)^2) = O(\log^{(k + 1)} n)$$ bits; and if a key $x$ is in a $j$-tiered \cubby for some $j < k$, then we can indicate which \cubby $x$ is in using $$O\left(\log (k - j) + \log \frac{(\log^{(j)} n)^2}{(\log^{(j + 1)} n)^2}\right) = O\left(\log (k - j) + \log^{(j + 1)} n\right) $$
bits. Since $k - j \le \log^* N - j = \Theta(\log^* \log^{(j)} n) = O(\log^{(j + 1)} n)$, the number of bits needed to indicate which \cubby $x$ is in can be upper-bounded by
$$O(\log^{(j + 1)} n).$$
In general, keys that are in lower-tiered \cubbys require more bits of metadata than those that are in higher tiers; but since there aren't very many low-tier keys, the total amount of metadata will remain small. (We'll see the full analysis of space consumption later in the section.)

\paragraph{Enforcing the invariants.}
We enforce the saturation invariant as follows. Whenever a deletion occurs in some non-tail \cubby $s$, we move one of the elements from the tail to that \cubby $s$. Whenever an insertion occurs in the \facility, we place the new element into the tail. Whenever the tail fills up, we create a new tail, and whenever the tail empties out, we eliminate that \cubby, and declare another one of the $ 1 $-tiered \cubbys to be the new tail. 

For each $ j\in [k -1] $, define  $t_j = \frac{(\log^{(j)} n)^2}{(\log^{(j + 1)} n)^2}$ to be the target number of $ j $-tiered \cubbys. This means that $ t_j $ is also the number of $ j $-tiered \cubbys whose aggregate size equals one $ j +1 $-tiered \cubby. Let $ r_j = K / (\log^{(j)} n)^2$ be the size of a $ j $-tiered \cubby. 

To enforce the distribution invariant, we must accommodate for the fact that new $ 1 $-tiered \cubbys are being added and removed over time. In general, for each $ j  \in [k - 1]$, whenever the number of $ j $-tiered \cubbys falls below $ t_j/ 2$, we take one of the $ j +1 $-tiered \cubbys and rebuild it as $ t_j $ \cubbys with tier $ j $ (this is a \defn{$ j $-creation rebuild}). And whenever the number of $ j $-tiered \cubbys rises above $3 t_j$, we take $ t_j $ \cubbys with tier $ j $ and rebuild them as a single $ j +1 $-tiered \cubby (this is a \defn{$j$-destruction rebuild}). 

We will describe how to deamortize these rebuilds (without compromising time or space efficiency) shortly. For now, let us simply observe that we only need to perform at most one $ j $-creation rebuild for every $\Theta(r_{j + 1})$ insertions that occur and we only need to perform at most one $ j $-destruction rebuild for every $\Theta(r_{j + 1})$ deletions that occur. By Lemma \ref{lem:buildicecube}, each $ j $-creation rebuild and each $ j $-destruction rebuild can be performed in time $\Theta (r_{j + 1}) $, with high probability in $ n $. It follows that, for each $ j $, the amortized time cost of the $ j $-rebuilds is $ O (1) $ per insertion/deletion. Since there are $k$ levels, the amortized time cost of all rebuilds is $ O (k) $ per operation.

When enforcing the saturation and distribution invariants, there is one technical subtlety that we must be careful about. Whenever we move an element from the tail \cubby to another \cubby, we should always choose that element at random\footnote{Note the it takes constant time to choose a random key in the tail \cubby, since the tail \cubby consists of only a $ O (1/\log n) $-fraction of the keys in the \facility, so we can afford to use an extra $\log n $ bits per key in the tail in order to maintain an auxiliary random-choice data structure that lets us select random keys. In fact, at any given moment, we should maintain a random-choice data structure for both the tail and $\Theta(1)$ other \cubbys in the same tier; this ensures that when one tail gets eliminated, another is ready to use. The constructions of the random-choice data structures are straightforward to deamortize to take $O(1)$ time per insertion/deletion.}; and whenever we perform a $ j $-destruction rebuild, we should partition the elements in the $ j +1 $-tiered \cubby randomly across the $ t_j $ new $ j $-tiered \cubbys being created. Call the random bits used to perform these choices  \defn{the non-hash randomness}.
Importantly, our use of non-hash randomness ensures that that for any given key $ x $, the choice of which \cubby it is currently in (within the \facility that it hashes to) is always a function exclusively of the sequence of operations that has been performed and of non-hash randomness, and it is \emph{not} a function of the hash functions used to perform insertions/deletions in \cubbys.

Finally, we must describe how to deamortize the $ j $-creation and $ j $-destruction rebuilds for all $ j \in [k - 1] $. A critical observation here is that all of the \facilitys have almost exactly the same sizes as each other at any given moment. Indeed, assuming that $ K $ is sufficiently large in $\polylog n $, then a Chernoff bound tells us that all of the \facilitys have the same number of elements as each other up to a factor of $ 1 \pm 1 / \polylog n$. Thus, we can synchronize the rebuilds for the \facilitys, so that whenever we perform a $ j $-creation or $ j $-destruction rebuild, we are actually performing a rebuild on all of the \facilitys at once over the course of $\Theta(r_{j + 1} N / K)$ operations. When this happens, we perform the rebuild on one \facility at a time (so it doesn't matter whether we perform the rebuild space efficiently); when we are performing a rebuild on a \facility, some insertions/deletions on that \facility may occur concurrently, but with high probability in $ n $ those operations will affect a total of $ O (1) $ distinct keys, and thus can easily be incorporated into the rebuild. Each $ j $-creation and $ j $-destruction takes total time $\Theta(r_{j + 1} N / K)$ across all \facilitys, and we only have to perform such a $ j $-creation or $ j $-destruction once every $\Theta(r_{j + 1} N / K)$ insertion/deletions on the hash table. Thus, for each $j$, we can spread out the work of performing $j$-creations/destructions to be $O(1)$ time per insertion/deletion. Summing over $ j \in [k - 1]$, this amounts to $ O (k) $ work per insertion/deletion.

This concludes the discussion of how to correctly enforce the two invariants without affecting space efficiency, and with only $ O (k) $ extra time being spent per insertion/deletion.

\paragraph{Analyzing space efficiency.} We first analyze the total space consumed by \cubbys, and then we analyze the total space consumed by the mini-arrays $D$ in each \facility.

Within a given \facility, all of the \cubbys are completely full except for the tail. The tail \cubby takes total space at most $O(K/\log n)$ machine words, which equals $ O (K) $ bits; thus the total space consumed by tails adds only $ O (1) $ bits per key in the hash table. The \cubbys that are full can be analyzed exactly as in Lemma \ref{lem:space}, allowing us to conclude that their total space consumption is $nw + O(n\log ^{(k)} n)$ bits.

Now let us analyze the space consumption of the mini-array $ D $ within a given \facility. For each key $ x $ we store which \cubby it is in. As discussed earlier in the section, if the key is in a $ j $-tiered \cubby for some $ j $, then it takes only $\Theta(\log^{(j + 1)} n)$ bits to encode which \cubby the key is in (note that this is always at most $ (\log\log n) $ bits, which means that it can be encoded as an $f$-value in a local query router). On the other hand, for each $j \in [k - 1]$, the fraction of keys that are in a $ j $-tiered \cubby is  $O(1 / (\log^{(j + 1)} n)^2)$. Thus the average size of the metadata in $D$ that we are storing for each key $ x $ is
\begin{align*}
 & O(\log^{(k + 1)} n) + \sum_{j < k} O\left(\frac{\log^{(j + 1)} n}{(\log^{(j + 1)} n)^2}\right) \\
& = O(\log^{(k + 1)} n) + O(1) \\
& = O(\log^{(k)} n)
\end{align*}
bits. By the same argument as in Lemma \ref{lem:space} (for bounding the total amount of space used by local query routers in \cubbys), we can deduce that, across the entire data structure, the total space used by mini-arrays $D$ is $O(n\log^{(k)} n)$ bits, with high probability in $ n $.

Thus, across the entire data structure, the total space consumption is 
$$nw + O(n \log^{(k)} n)$$
bits, as desired. 

\paragraph{Supporting large changes in size.}
So far we have focused exclusively on the case where the average size of each \cubby stays within the range $ [K, 2K]$.

We can generalize this to support a larger range of sizes with the following approach. Every time that the hash table's size changes by a constant factor, we move all of the elements from the current hash table $H_1$ into a new hash table $H_2$ whose capacity is twice as large (resp. small) as that of $H_1$. This means that each \facility (resp. pair of adjacent \facilitys) in $H_1$ becomes a pair of adjacent \facilitys (resp. single \facility) in $H_2$. The transformation from $H_1$ to $H_2$ takes $O(nk)$ time, and can be spread across $\Theta(n)$ insertions/deletions to take $O(k)$ time each. The transformation can also be performed space efficiently, by transforming one \facility (resp. one pair of \facilitys) at a time, so that each key/value pair only takes up space in one of the two hash tables.

Putting the pieces together, we arrive at the following theorem:
\begin{theorem}
One can construct a dictionary storing $w$-bit key/value pairs so that if $n$ is the current number of keys and $k\in [\log^*n]$, insertions/deletions take time $ O (k) $, queries take time $O(1)$, and, if $\log n = \Theta(w)$, the total space consumption is $wn + O(n \log^{(k)} n)$ bits. The running-time and space guarantees are with high probability in $n$.  
\end{theorem}

\subsection{Succinctness through quotienting}\label{sec:quotient}

In this section, we modify our hash table so that it can store $n$ keys from a polynomial-size universe $U$ in total space
$$\log\binom{|U|}{n} + \Theta\left(n \log^{(k)} n\right) = n \log \binom{|U|}{n} +  \Theta\left(n \log^{(k)} n\right)$$
bits, where $\log \binom{|U|}{n}$ is the information-theoretical lower bound on the number of bits needed to store $ n $ elements from a universe of size $|U|$.
If we are also storing $\lambda$-bit values for some $\lambda = O(\log n)$, then our total space consumption becomes
$$n \log \binom{|U|}{n} + n \lambda +  \Theta\left(n \log^{(k)}\right)$$
bits.

\paragraph{A recap: the current structure of our hash table.}
Before we begin, let us briefly recap the structure of our hash table, and give names to the components and hash functions that are used.
At any given moment, use $N$ to denote the power-of-two range $[N, 2N]$ that the current table-size $n$ is in. 

The layout of our hash table is as follows: for some parameter $K = \polylog N$, there are $N / K$ \facilitys, each of which contains a metadata mini-array $D$ of size $K$. Call the metadata array in the $i$-th \facility $D_i$, and think of the $D_i$s as partitioning a larger array $D^*$ with $N$ entries. Each key $x$ hashes to a random \facility $i$, and then to a random slot in that \facility's array $D_i$---equivalently, this means that each key $x$ hashes to a random slot $g^*(x)$ in the array $D^*$. 

Once a key $x$ selects a \facility, the key $ x $ is then assigned to one of the \cubbys in the $ i $-th \facility (and the local query router $D^*[g^*(x)]$ stores which \cubby $x$ is assigned to). 

Each \cubby $I$, capable of storing $|I|$ keys, has the following layout. The \cubby maintains two metadata mini-arrays $A_I$ and $M_I$ and a storage array $S_I$ of size $|I|$; the mini-array $A_I$ stores local query routers, the mini-array $M_I$ stores metadata used to implement insertions/deletions in constant time, and the array $S_I$ stores the actual key/value pairs. Each key $x$ stored in the \cubby hashes to some target position $g_I(x) \in [|I|]$,  and this target position is used both to determine (a) which query-router in $A_I$ handles key $x$, and (b) what $x$'s target position is in the $k$-\kicktree used to determine where keys go in $S_I$. 

It turns out that, in this section, it will be important that the hash function $g_I$ reuses the random bits from $g^*$, so $g_I(x)$ is the lowest-order $\log |I|$ bits of $g^*(x)$.
As a convention, when discussing the bits of a number $a$, we will use $a[i, j]$ to refer to bits ranging from the $i$-th least significant bit to the $j$ least significant bit, for $i \le j$. Thus $g_I(x) = g^*(x)[1, \log|I|]$. Also, using the same notation, if $r$ is the facility that a key $x$ is in, then $r = g^*(x)[\log K + 1, \log N]$. 

\paragraph{Modifying the data structure to use quotienting.}
We can make our data structure more space efficient by using the \defn{quotienting} technique \cite{KnuthVol3}. Let $\pi$ be a random permutation hash function on the universe $U$.\footnote{Recall that in Section \ref{sec:prelim}, we discussed how to simulate permutation hash functions that are $\poly(n)$-independent.} By applying $\pi$ to each key that we are storing, we can assume that the keys being stored form a random subset of $U$.

Since the keys $x$ are random, we can define $g^*(x)$ to simply be the first $\log N$ bits of $x$, that is, $g^*(x) = x[1, \log N]$. We then make two modifications to our data structure: first, we modify each \cubby so that its storage array stores only the final $\log U - \log N$ bits $x[\log N + 1, \log |U|]$ of each key $x$ (as well as the $\lambda$-bit value stored with the key); second we add additional metadata (which we will describe in a moment) so that, for each key $x$, we can recover the first $\log N$ bits of $x$ despite the fact that we are not explicitly storing them.

The metadata that we add is the following. For each \cubby $ I $, we add a new mini-array $ B_I $ of size $|I|$. If $S_I[i]$ is empty (there is no key stored there), then $B_I[i]$ stores nothing. Otherwise, if $x$ is the key being stored in $S_I[i]$, then $B[i]$ stores the following two quantities: (a) the difference $g_I(x) - i$ and (b) the $\log (K / |I|)$ bits $x[\log |I| + 1, \log K]$ of $g^*(x)$. The first quantity can be added to $i$ to reconstruct $g_I(x) = x[1, \log |I|]$. The second quantity can be appended to $g_I(x)$ to obtain $g(x)[1, \log K] = x[1, \log K]$. And finally, the \facility number can be used to recover the final $\log N - \log K$ bits of $g^*(x)$ (i.e., if we are in facility $r$ then $x[\log K + 1, \log N] = g^*(x)[\log K + 1, \log N] = r$). This means that we can reconstruct the entire quantity $g^*(x) = x[1, \log N]$, which are the bits of $x$ that we do not explicitly store.

To analyze our new data structure, there are two tasks that we must complete. The first is to analyze the amount of space used by the $B_I$'s. 
The second is to handle the fact that the hash function $g^*(x)$ is no longer a fully independent hash function, and instead contains small negative correlations between keys (by virtue of being a permutation hash function).

\paragraph{Analyzing space consumption.}
Let $ R $ be the total number of bits used to store the metadata mini-arrays $B_I$ for each \cubby $I$. Then the total space consumed by our data structure is
$$ n (\log N - \log U + \lambda) + O(R + n \log^{(k)} n) = n \log \binom{|U|}{n} + n \lambda + O\left(R + n \log^{(k)} n\right).$$
Thus, we wish to show that $R \le O( n \log^{(k)} n)$. To prove this, we will show that the metadata stored by the $B_I$'s takes at most as much total space as the metadata stored in other mini-arrays in the data structure. Since we have already bounded the space consumption of those other mini-arrays by $O(n \log^{(k)} n)$ bits (with high probability), it follows that the same bound applies to the $B_I$'s.

For each key $ x $ in some position $i$, $B_I$ stores the quantity $g_I(x) - i$. Notice, however, that the number of bits needed to store this quantity is simply the probe complexity of $x$ in the \cubby. The same number of bits are already stored in $A_I$ in order to route queries.

The second quantity that $ B_I $ stores for $x$ is $\log (K / |I|)$ bits of $g^*(x)$. If $|I|$ has size $K / \log^{(j)} n$, then this quantity is $\Theta(\log^{(j + 1)} n)$ bits. But that's the same number of bits that are already being stored for $x$ in the metadata mini-array $D^*$. 

Thus the total space consumed by our hash table is
$$n \log \binom{|U|}{n} + n \lambda + O\left(n \log^{(k)} n\right)$$
bits.

\paragraph{Proving correctness in the face of negative dependencies.}
Our final task is to handle the fact that the hash function $ g ^*$ is not a fully independent hash function. Recall that we generate $g^*$ by first performing a random permutation $\pi$ on the universe (so each $x \in U$ is mapped to $\pi(x)$), and then setting $g^*(x)$ to be $\pi(x) \mod N$. The problem is that the $\pi(x)$'s form a random subset of $U$ \emph{without replacement}. This means that the quantities $\pi(x)$ are necessarily unique, which prevents them from being totally independent from one another.

There are three places in our analysis where we must be careful:
\begin{itemize}
\item For the $k$-\kicktree analysis from Theorem \ref{thm:probeupper} to apply to each \cubby, we need to be able to apply a Chernoff bound to the number of keys $x$ in the \cubby that have $g_I(x)$ in any given sub-interval $J \subseteq [|I|]$. 
\item In order so that each local query router in $D^*$ (and each local query router in each $A_I$) stores metadata for at most $O(\log n / \log \log n)$ keys, we need to be able to apply a Chernoff bound to the number of keys that hash to a given local query router.
\item So that the \facilitys can all be resized simultaneously, we need each \facility to contain at most $(1 + 1 / \polylog n) K/n$ keys, at any given moment, with high probability in $n$. In other words, we need to be able to apply Chernoff bounced the number of keys $x$ that hash to a given \facility.
\end{itemize}

In summary, for any subset $Q \subseteq [N]$, we need to be able to apply a Chernoff bound to the number of keys $x$ satisfying $g^*(x) \in Q$. Or, more generally, for any subset $Q \subseteq U$, we need to be able to apply a Chernoff bound to the number of keys $x$ that satisfy $\pi(x) \in Q$.\footnote{Technically, we also want to be able add conditions to this as follows: we want to be able to apply a Chernoff bound to the number of keys $\pi(x) \in Q$, having already conditioned on which keys $x \in U$ satisfy $\pi(x) \in Q'$ for some $Q' \supset Q$. (This lets us split each \cubby into $\polylog n$ chunks that are analyzed independently.) These conditioned Chernoff bounds can be achieved in the same way as the condition-free Chernoff bounds.} 

Fortunately, the indicator random variables $Q_x$ indicating whether $\pi(x) \in Q$ are negatively correlated. That is, for any subset $R \subseteq [U]$ we have
$$\Pr\left[\prod_{x \in R} Q_x = 1\right] \le \prod_{x \in R} \Pr[Q_x = 1] = \left(\frac{|Q|}{|U|}\right) ^{|R|}.$$
Chernoff bound are known to hold for indicator random variables satisfying this type of negative correlation (see, e.g., Section 1.3.1 of \cite{impagliazzo2010constructive}). Thus, our high-probability analysis from the previous sections continues to hold without modification.

Putting the pieces together, we arrive at the following theorem:
\begin{theorem}
Suppose we wish to store key/value pairs where keys are from a universe $U$, and values are $\lambda \le O(\log |U|)$ bits. Assume a machine-word size of $\Omega(\log |U|)$ bits and let $k \ge 0$.

One can construct a dictionary that supports insertions/deletions in time $ O (k) $, that supports queries in time $O(1)$, and that offers the following guarantee on space: if the current number of keys is $n$, and $\log |U| = \Theta(\log n)$, then the total space consumption is
$$n \log \binom{|U|}{n} + n \lambda + O\left(n \log^{(k)} n\right)$$
bits. The running-time and space guarantees are with high probability in $n$.  
\label{thm:dictionary}
\end{theorem}

\section{Large Keys, Small Keys, and Filters}

In this section, we give three extensions of Theorem \ref{thm:dictionary}. In Subsection \ref{sec:largekeys}, we consider the setting in which key/value pairs are $w = \omega(\log n)$ bits, and we show that the guarantee of $O(\log^{(k)} n)$ wasted bits per key can be extended to any $w \in n^{o(1)}$.  In Subsection \ref{sec:small}, we consider the setting in which key/value pairs are very small, taking a total of $\log n + s$ bits for some $s = o(\log n)$. We show that, if $s$ is even slightly sublogarithmic, then it is possible to reduce the number of wasted bits per key all the way to $o(1)$. Finally, in Subsection \ref{sec:filters}, we apply our results to the problem of constructing space-efficient dynamic filters.

Our results on very-large and very-small keys both hinge on novel reductions that transform the the very-large/very-small case into the standard $\Theta(\log n)$-bit case. These reductions may be independently useful in future work.

\subsection{Supporting Large Keys/Values}\label{sec:largekeys}

So far we have restricted ourselves to the case where keys/values are $ O (\log n) $ bits each. We prove the following theorem:

\begin{theorem}
Suppose we wish to store key/value pairs where keys are from a universe $U = [2^a]$, values are $b$ bits, and $a + b \le n^{o(1)}$. Assume a machine-word size of $\Omega(a + b)$ bits and let $k \ge 0$.

One can construct a dictionary that supports insertions/deletions in time $ O (k) $, that supports queries in time $O(1)$, and that offers the following guarantee on space: if the current number of keys is $n$, then the total space consumption is
$$n \log \binom{|U|}{n} + n b + O\left(n \log^{(k)} n\right)$$
bits. The running-time and space guarantees are with high probability in $n$.  
\label{thm:largekeys}
\end{theorem}

Historically, the task of supporting universes $ U $ of superpolynomial size has proven to be quite difficult. Indeed, the best known guarantees for worst-case constant-time hash tables \cite{arbitman2010backyard, liu2020succinct} use
$$\log \binom{|U|}{n} + \Omega(\min(|U|^{\Omega(1)}, n \log n))$$
bits of space,\footnote{The result of Liu et al. \cite{liu2020succinct} uses $\log \binom{|U|}{n} + |U|^{\Omega(1)}$ bits of space and the result of Arbitman et al. \cite{arbitman2010backyard} uses $\log \binom{|U|}{n}  + \Omega(n \log n)$ bits of space. The solution that Arbitman et al. use is to just hash elements to a polynomial-size universe---if one is willing to allow for query-correctness to be violated a $1 / \poly n$ fraction of the time, then it suffices to store only these hashes, but otherwise one must also store the entire original element.} and the best known guarantee for constant expected-time hash tables \cite{raman2003succinct} uses  
$$\log \binom{|U|}{n} + \Omega(n \log (a + b))$$
bits of space.  In contrast, Theorem \ref{thm:largekeys} says that, as the size of the universe grows, the number of wasted bits per key does not---it remains $O(\log^{(k)} n)$ bits per key even for very large keys/values.

There are several reasons that large universes are difficult to handle. One difficulty is that it is not known how to efficiently perform quotienting on keys from large universes. For any universe $U$, it is known how to construct an efficient family of $n^{\epsilon}$-wise $(1 / \poly(n))$-dependent permutations \cite{naor1999construction, luby1988construct, kaplan2009derandomized}---but each member $\pi \in \Pi$ requires $|U|^{\Omega(1)}$ bits to store, so if $|U|$ is super-polynomial, then we cannot store $\pi$ in polynomial space. This has prevented past work \cite{liu2020succinct, arbitman2010backyard, bender2021all} from using quotienting on large universes. Another difficulty \cite{liu2020succinct, arbitman2010backyard, bender2021all} with large universes is that we can no longer afford to store lookup tables of size $|U|^{\Omega(1)}$---this is a serious bottleneck for any hash table that uses the Method of Four Russians (including the hash tables in this paper) as a path to worst-case constant-time operations. Finally, even if both of these issues were eliminated, the known techniques \cite{arbitman2010backyard, liu2020succinct, raman2003succinct} for constructing succinct hash tables would still incur an asymptotic blow-up in wasted-bits-per-key as the universe size grows.

In this section, we present a new approach for handling large universes that lets us get around all of the above issues at once. Define an \defn{$(a, b)$-dictionary} to be a dictionary that stores $a$-bit keys with $ b $-bit values. We prove that there is an efficient (high probability) reduction from the problem of constructing a succinct $(a, b)$-dictionary to the problem of constructing a succinct $(\Theta(\log n), a + b - \Theta(\log n))$-bit dictionary. This means that we can always assume that keys are $\Theta(\log n)$ bits. Once we have proven this reduction, our only challenge will be to handle large values, which it turns out will be relatively straightforward using the techniques we have already developed.

The reduction is captured in the following theorem. 
\begin{theorem}
Let $a, b, N, \gamma \in \mathbb{N}$ be parameters, and let $\gamma > 0$ be a sufficiently large constant. Suppose that $a \ge \gamma \log N$, and suppose that the machine-word size $w$ satisfies $w \ge \Omega(a + b)$. Finally let $f_N(n)$ be a non-negative non-decreasing function.

Suppose we have a dynamically resizable $(\gamma \log N, a + b - \gamma \log N + 1)$-dictionary $S$ that is capable of storing $n$ key-value pairs in space $n (a + b - \log n) + f_N(n)$ bits for any $n \in [N]$, while offering running-time guarantees that are high-probability in $N$. Then we can construct a dynamically resizable $(a, b)$-dictionary $L$ that is capable of storing $n$ key-value pairs in space $n (a + b - \log n) + f(n) + O(\sqrt{N})$ bits, for any $n \in [N]$, and that, with high probability in $N$, takes at most $O(1)$ more time per operation than does $S$. 
\label{thm:reduction}
\end{theorem}
\begin{proof}
To capture the fact that $\gamma$ is at least a sufficiently large positive constant, we shall treat $\gamma$ as an asymptotic variable. We implement $L$ with three data structures: 
\begin{itemize}
\item The first data structure is the $(\gamma \log N, a + b - \gamma \log N + 1)$-dictionary $S$;
\item The second data structure is a dynamically resizable $(\sqrt{\gamma} \log N, \log N)$-dictionary $D_1$ that can store $r$ keys/value pairs in $O(r \sqrt{\gamma} \log N + \sqrt{N})$ bits of space for any $r \le N$, and that supports constant-time operations with high probability in $N$;
\item The third data structure is a $(\sqrt{\gamma} \log N, a + b - \gamma \log N + \sqrt{\gamma}\log N)$-dictionary $D_2$ that can store $r$ key/value pairs in $\Theta(r \sqrt{\gamma} \log N + \sqrt{N}) + r(a + b  - \gamma \log N + \sqrt{\gamma}\log N)$ bits of space for any $r \le N$, and that supports constant-time operations with high probability in $N$.
\end{itemize}

Let us briefly comment on how to implement $D_1$ and $D_2$. Notice that $D_1$ does not need to be succinct (or even compact)---it can be implemented using any standard constant-time dictionary that supports load factor $\Omega(1)$. To implement $D_2$, we store values with a layer of indirection, meaning that the dictionary allocates separate $(a + b  - \gamma \log N + \sqrt{\gamma}\log N)$-bit chunks of memory for each individual value, and stores a $\Theta(\log N)$-bit pointer to that value. The $(\sqrt{\gamma} \log N, O(\log N))$-dictionary that stores the keys and the pointers can again be implemented with any standard constant-time dictionary that supports load factor $\Omega(1)$.\footnote{We remark that the $\sqrt{N}$ term in the space-usage for $D_1$ and $D_2$ stems from the fact that, in order for the probabilistic guarantees of $D_1$ and $D_2$ to be high-probability in $N$, we need $D_1$ and $D_2$ to be size at least $\poly N$.}

For any key $x \in [2^a]$, define the \defn{core} $\phi(x)$ to be the first $\gamma \log N$ bits of $x$, and define $\psi(x)$ to be the final $|x| - \gamma \log N$ bits of $x$. So $\phi(x) \circ \psi(x) = x$. 

We will use $S$ to store key/value pairs of the form $(\phi(x), \psi(x) \circ y \circ \ell)$ where $\ell \in \{0, 1\}$ is an extra bit of information that we call the \defn{abundance bit}. If there is exactly one key $ x \in L$ with a given core $\phi (x) $, then $ S $ stores $(\phi (x),\psi(x) \circ y \circ 0)$. If there is more than one key $x$ with a given core $j = \phi(x)$, then $S$ stores $(\phi(x), \psi(x) \circ y \circ 1)$ for \emph{one} such key/value pair $(x, y)$. The fact that the abundance bit is set to $1$ in this latter case indicates that there are also other key/value pairs $(x', y')$ satisfying $\phi(x') = \phi(x)$. In this case, we say that the core $\phi(x)$ is \defn{rabid}, and any key $x' \neq x$ that has core $\phi(x') = \phi(x)$ is also said to be a \defn{rabid key} (regardless of whether $x'$ is a member of the dictionary $L$ that we are constructing).

Rabid keys $z \in L$ (and their corresponding values $w$) are stored as follows. Let $h (\phi(z)) \in [\sqrt{\gamma} \log N]$ be a pairwise-independent hash, and let $g(z) \in [\sqrt{\gamma} \log N]$ also be a pairwise-independent hash. We store the pair $(g(z), h(\phi(z)) \circ \psi(z) \circ w)$ in $D_2$. If we ever attempt to insert a key $g(z)$ into $D_2$ that is already there (i.e., we have a collision), then we rebuild the entire data structure with new random bits (this only occurs with probability $1 / \poly N$ per insertion).  Finally, for each rabid core $\phi(z)$, we store the pair $(h(\phi(z)), q)$ in $D_1$, where $q$ is a reference counter keeping track of the number of rabid keys in $L$ have that core. 

The data structure $D_1$ has two purposes. The first (and less important) purpose is to maintain the reference counter $q$ so that, on deletions, we know when there are no longer any rabid elements with a given core $\phi(z)$, at which point we can both remove  $h(\phi(z))$ from $D_1$ and we can set the abundance bit for $\phi(z)$ in $S$ to be $0$. 

The more important purpose of $ D_1 $, however, is to detect collisions between $h(\phi(z))$ and $h(\phi(z'))$ for rabid keys $z, z'$ satisfying $\phi(z) \neq \phi(z')$. Whenever we insert some rabid key $z$, we can tell based on the abundance bit for $\phi(z)$ (prior to the insertion) whether this is the first rabid key in $L$ to have core $\phi(z)$; if it is the first such rabid key, but there is already a pair of the form $(h(\phi(z)), q)$ in $D_1$, then that means a collision on $h$ has occurred, and we rebuild the entire data structure from scratch. Such collisions occur with probability only $1 / \poly N$ per insertion.

What $ D_1 $ ensures is that $h$ is injective on the set of rabid cores---that is, if $\phi(z)$ and $\phi(z')$ are two different rabid cores, then $h(\phi(z)) \neq h(\phi(z'))$. It follows that there is a bijection between rabid keys $z \in U$ and pairs $(h(\phi(z)), \psi(z))$. Thus $D_2$ can be used to perform queries on rabid keys $z$ as follows: check if there is a key-value pair of the form $(g(z), h(\phi(z)) \circ \psi(z) \circ w)$; if there is, then return value $w$, and otherwise declare that $z$ is not present. 

It is straightforward to perform insertion/deletion/queries using $S, D_1, D_2$. Note that, if a non-rabid key $x$ is deleted from $L$, but there is a rabid key $x'$ in $L$ that has the same core $\phi(x') = \phi(x)$, then we must move one such rabid key $x'$ out of $(D_1, D_2)$ and into $S$ (so that $x'$ is no longer rabid). 

Our final task is to analyze the space consumption of our data structure $L$. We will use $r$ to denote the number of rabid keys in $L$, and we will use $k$ to denote the number of non-rabid keys in $L$ (so $n = k + r$). Since $\gamma$ is a sufficiently large constant, the data structures $ D_1 $ and $ D_2 $ collectively use
$$O(\sqrt{N}) + O(r \sqrt{\gamma} \log N) + r (a + b - \gamma \log N + \sqrt{\gamma}\log N)$$
bits of space, where the final term accounts for the space consumed by the values of $D_2$, and the first two terms account for the space consumed by the rest of $D_1, D_2$. Using the fact that $\gamma $ is a sufficiently large positive constant, it follows that $ D_1 $ and $ D_2 $ collectively use
$$O(\sqrt{N}) + r (a + b - \log n)$$
bits. The data structure $ S $, on the other hand, uses
$$k (a + b - \log n) + f(k) \le k (a + b - \log n) + f(n)$$
bits. The total number of bits used is therefore 
$$n (a + b - \log n) + f(n) + O(\sqrt{N}).$$
\end{proof}

The result of Theorem \ref{thm:reduction} is that we can always assume without loss of generality that our keys are size $O(\log n)$ bits. Thus, in order to prove Theorem \ref{thm:largekeys}, it suffices to construct an $(O(\log n), b)$-dictionary, where the only constraint on $b$ is $b \le n^{o(1)}$ (and where the machine-word size $w$ satisfies $w \ge \Omega(\log n + b)$). 

\paragraph{Storing large values space efficiently.}
To store large values, we exploit an interesting feature of the dynamically resizable dictionary that we constructed for the proof of Theorem \ref{thm:dictionary}: in each \facility, all of the \cubbys except for the tail are \emph{completely full}. Thus, for each \cubby $I$ (except for the tail), we can allocate $bI$ bits of space $V_I$ to store values---the key stored in the $j$-th position of $I$ has its value stored in bits $(j - 1)b + 1, \ldots, jb$ of $V_I$. Importantly, $V_I$ is fully saturated, so it wastes no space.

To handle the keys/values in the tail, recall that the tail consists of less than an $O(\frac{1}{\log n})$-fraction of the keys in the \cubby. Thus, we can individually allocate space for each value in the tail, and we can store a $\Theta(\log n)$-bit pointer to that value. The $\Theta(\log n)$-bit pointers contribute only $O(1)$ amortized bits of space overhead per key in the data structure.

One technical detail that we must be careful about is that, whenever an \cubby toggles between being/not-being the tail, we must change how the values are stored in that \cubby. This is straightforward to do in a deamortized fashion using the same deamortized-rebuilding techniques as in Section \ref{sec:variable}.

In summary, we have the following lemma:

\begin{lemma}
Suppose we wish to store key/value pairs where keys are from a universe $U = [2^a]$, values are $b$ bits. Suppose $a = O(\log n)$ and $b \le n^{o(1)}$. Assume a machine-word size of $\Omega(a + b)$ bits. Finally, let $k \ge 0$.

One can construct a dictionary that supports insertions/deletions in time $ O (k) $, that supports queries in time $O(1)$, and that offers the following guarantee on space: if the current number of keys is $n$, then the total space consumption is
$$n \log \binom{|U|}{n} + n b + O\left(n \log^{(k)} n\right)$$
bits. The running-time and space guarantees are with high probability in $n$.  
\label{lem:largevalues}
\end{lemma}

Combined, Theorem \ref{thm:reduction} and Lemma \ref{lem:largevalues} imply Theorem \ref{thm:largekeys}.

\subsection{Optimizing for very small keys}\label{sec:small}

In this section we consider the case of very small keys, that is keys of size $\log n + o(\log n)$ bits. For most of the section we shall focus exclusively on dictionaries that store keys without values, but at the end of the section we will also generalize to the case where the dictionary also stores very small values.

We begin with the fixed-capacity case. We show that, if keys are of size $\log n + s$ for some $s \le \log n / \log \log \cdots \log n$ (where there are a constant number of logarithms), then it is possible to construct a constant-time dictionary with $o(1)$ wasted bits per key.

\begin{theorem}
Let $k, n$ be parameters, where $k\in [\log^* n]$. Let $\phi = \Theta(\log^{(k)} n)$. Let $ U = [2^{\log n + s}]$ for some $s$ satisfying $s \in \omega(1) \cap o(\log n)$ and suppose that $s \le O(\log n / \phi)$.

There exists a fixed-capacity dictionary that stores up to $n$ keys from $U$ at a time, that supports insertions/deletions in time $O(k)$ (with high probability in $n$), that supports queries in time $O(1)$, and that uses a total of 
$$\log \binom{|U|}{n} + o(n)$$
bits of space (with high probability in $n$).
\label{thm:small}
\end{theorem}

We will assume without loss of generality that $\phi \le \log \log n$ and that $\phi$ is rounded to the nearest power of two. We will also assume without $\phi = \omega(1)$, since the fact that $s = o(\log n)$ ensures that the theorem requirements hold for some $\phi = \omega(1)$. And finally, we will assume without loss of generality that $s \le o(\log n / \phi)$, since we otherwise can replace $\phi$ with $\phi' = \Theta(\log^{(k + 1)} n)$ and prove the theorem for the new $\phi'$. 

We begin by describing the data structure. First, we assume that the keys form a random subset of $U$; as noted in Section \ref{sec:prelim}, it is known how to construct permutation hash functions that simulate this assumption while preserving time and space guarantees.

We use the first $\log (n / \phi)$ bits of each key to assign it to a random one of $n / \phi$ bins $R_1, \ldots, R_{n / \phi}$. We maintain an array $ A $ of $n / \phi$ $O(\log \phi)$-bit counters $A_1, \ldots, A_{n / \phi}$, where each $A_i$ is always in the range $[0, 100\phi]$. Whenever we insert an element $x$ that maps to some bin $R_i$, we examine the counter $A_i$. If $A_i < 100 \phi$, then we declare $x$ to be \defn{standard} and we increment $A_i$; otherwise, we declare $x$ to be \defn{non-standard}, and we leave $A_i$ unchanged. Similarly, we decrement $A_i$ whenever we delete a standard element $x$ that belongs to bin $R_i$, but we do not decrement $A_i$ when we delete a non-standard element.

We take different approaches to storing standard versus non-standard elements. Non-standard elements are stored in a secondary \defn{backyard hash table} $B$, constructed via Theorem \ref{thm:dictionary} to incur $O(\log^{(k + 1)} n)$ wasted bits per key. The large number of wasted bits per key is okay because only a small number of elements will reside in $B$.

Standard elements, on the other hand, are stored as follows. For any given bin $R_i$, we encode the set of standard elements that reside in that bin using a single $o(\log n)$-bit integer $E_i$ (we will describe how to construct $E_i$ later). Importantly, the number of bits used for $E_i$ is a strict function of the number $A_i$ of elements being encoded. 

Since different $E_i$s have different sizes, we cannot store them contiguously in an array. Instead, we again make use of Theorem \ref{thm:dictionary}. We maintain $100 \phi$ dynamically-resized hash tables $ H_1,\ldots, H_{100 \phi}$, each of which is parameterized to incur $O(\log^{(k + 1)} n)$ wasted bits per key. For each bin $ R_i $, we store the key-value pair $(i, E_i)$ in hash-table $H_{A_i}$. Notice that this construction ensures that each hash table $ H_i $ storing fixed-size keys and values. Also notice that, even though the $H_i$s incur a relatively large number of wasted bits per key, there are only $O(n / \phi)$ total elements in the $H_i$s, one for each of the $n / \phi$ bins. 

To complete the description of the data structure, we show that it is possible to encode each $E_i$ space efficiently.
\begin{lemma}
For any bin $R_i$, the set of standard keys in that bin can be encoded using $\log \binom{\phi 2^s}{A_i} + O(1) = o(\log n)$ bits. Moreover, the encodings can be updated/queried in constant time using $O(\sqrt{n})$ bits of metadata.
\label{lem:Ri}
\end{lemma}
\begin{proof}
Recall that keys are $\log n + s$ bits. All of the keys in $R_i$ agree on their first $\log (n / \phi)$ bits, so they differ in only their final $s + \log \phi$ bits. The number of possibilities for the $A_i$ keys encoded by $E_i$ is 
$$\binom{\phi 2^s}{A_i}.$$

If we can show that this is $2^{o(\log n)}$, then the lemma follows from the Method of Four Russians (see discussion in Section \ref{sec:metadata}). To complete the proof, observe that
\begin{align*}
     \binom{\phi 2^s}{A_i} & \le  \binom{(\log \log n) 2^{o((\log n / \phi))}}{O(\phi)} \\
                                         & \le \log \left((\log \log n) 2^{o(\log n / \phi)}\right)^{O(\phi)} \\
                                         & \le (\log \log n)^{O(\log \log n)} 2^{o(\log n)} \\
					&  \le \log 2^{o(\log n)}.
\end{align*}
\end{proof}

We now proceed to bound the space consumption of the hash table. We begin with a simple approximation for binomial coefficients.

\begin{lemma}
For all $a = \omega(b)$, we have
$$\log \binom{a}{b} = b \log a - b \log b + b \log e \pm o(b),$$
and for all $a \ge b$, we have
$$\log \binom{a}{b} \le b \log a - b \log b + b \log e + o(b).$$
\label{lem:binom}
\end{lemma}
\begin{proof}
If $a = \omega(b)$, then
\begin{align*}
    \log \binom{a}{b} & = \log \left( \frac{a \cdot (a - 1) \cdots (a - b + 1)}{b!}\right) \\
                        & = \log \left( \frac{a^{b}}{b!}\right) \pm o(b) \\
                        & = b \log a - \log (b!) \pm o(b) \\
                        & = b \log a - b \log b + b \log e \pm o(b),
\end{align*}
where the final step follows from Stirling's inequality.  By a similar sequence of arguments, if $a \ge b$, then  
\begin{align*}
    \log \binom{a}{b} & = \log \left( \frac{a \cdot (a - 1) \cdots (a - b + 1)}{b!}\right) \\
                        & \le \log \left( \frac{a^{b}}{b!}\right) \\
                        & = b \log a - \log (b!) \pm o(b) \\
                        & = b \log a - b \log b + b \log e \pm o(b).
\end{align*}
\end{proof}

Next we bound the number of elements in the backyard, at any given moment.
\begin{lemma}
With high probability in $n $, the number of non-standard elements is $O(n / 2^\phi)$ at any given moment.
\label{lem:backyardbound}
\end{lemma}
\begin{proof}
Break the bins $R_1, R_2, \ldots, R_{n / \phi}$ into $m = (n / \phi) / \polylog n$ collections $C_1, \ldots, C_{m}$, each consisting of some number $\ell = \polylog n$ of bins. The assignments of keys to bins are negatively correlated, so we can use a Chernoff bound for negatively correlated random variables \cite{impagliazzo2010constructive} to deduce that, with high probability in $n$, the number of keys assigned to any given collection is at most $2 \ell \phi$ at any given moment. 

Let $a_1, \ldots, a_q$ be the set of keys currently present. By a union bound, we have that with high probability in $n$, every chunk $C_i$ had at most $2\ell \phi$ keys assigned to it during each of the time steps in which $a_1, \ldots, a_q$ were inserted. Condition on this being the case, and further condition on which specific keys hash to which  $C_i$s.

Define $c_1, \ldots, c_m$ so that $c_i$ is the number of non-standard keys currently in $C_i$. Having conditioned on which keys hash to which collections, the $c_i$s are independent. We will show that $\E[c_i] = O(\ell \phi / 2^\phi)$, meaning that $\E[\sum_i c_i] = O(n / 2^\phi)$. Since each $c_i$ is guaranteed to be at most $2 \ell \phi = \polylog n$, we can apply a Chernoff bound to the sum $\sum_i c_i$ to deduce that the total number of non-standard elements is $O(n / 2^\phi)$ with high probability.

We conclude the proof by establishing that $\E[c_i] = O(\ell / 2^\phi)$. By linearity of expectation, it suffices to show that any given key $x$ has a $O(1 / 2^\phi)$ probability of being non-standard. This follows by applying a Chernoff bound (again for negatively correlated random variables) to the number of keys that hash to $x$'s bin. When $x$ is inserted, there are at most $2\ell\phi$ keys in $x$'s collection, each of which has a $1 / \ell$ probability of hashing to $x$'s bin, so by a Chernoff bound we have that the probability of there being already $100 \phi - 1$ keys in $x$'s bin is at most $1 / 2^\phi$. This completes the proof.
\end{proof}

We can now bound the total space consumed by the hash table.
\begin{lemma}
The total space consumed by the hash table is $\log \binom{n2^r}{n} + o(n)$ bits, with high probability in $n$.
\label{lem:boundspacelittle}
\end{lemma}
\begin{proof}
The array $A$ of counters uses $O((n / \phi) \log \phi) = o(n)$ bits. Let $J_0$ be the number of items in the backyard. Then the backyard uses space
$$\log \binom{n2^s}{J_0} + O(J_0\log^{(k + 1)} J_0)$$ 
bits. By Lemma \ref{lem:backyardbound}, we have $J_0 \le O(n / 2^\phi) = o(n / \log^{(k + 1)} n)$, so the total space consumed by the backyard is at most
$$\log \binom{n2^s}{J_0} + o(n)$$
bits. 

For $i \in [n]$, let $J_i$ be the number of elements in each $H_i$. Notice that $\sum_{i = 1}^n J_i = n / \phi$. Since $H_i$ stores $\log n$-bit keys, stores $\log \binom{\phi 2^s}{A_i} + O(1)$-bit values ,and wastes $O(\log^{(k + 1)} n)$ bits per key, the total space taken by a given $H_i$ is 
$$O(J_i \log^{(k + 1)} n) + \log \binom{n}{J_i} + J_i \log \binom{\phi 2^s}{A_i}$$
bits.\footnote{Here we ignore the $n^{1 - \Omega(1)}$ total bits used for Method of Four Russians and for storing hash functions.}  

Putting the pieces together, the total space consumed by the hash table is
\begin{align*}
\sum_{i = 1}^{100 \phi} O(J_i \log^{(k + 1)} n) + \sum_{i = 1}^{100 \phi} \log \binom{n}{J_i} + \sum_{i = 1}^{n / \phi} \log \binom{\phi 2^s}{A_i} + \log \binom{n2^s}{J_0} + o(n). \\
\end{align*}
Since $\sum_i J_i = n / \phi = o(n / \log^{(k + 1)} n)$, this is at most
\begin{align*}
\sum_{i = 1}^{100 \phi} \log \binom{n}{J_i} + \sum_{i = 1}^{n / \phi} \log \binom{\phi 2^s}{A_i} + \log \binom{n2^s}{J_0} + o(n). \\
\end{align*}
Applying Lemma \ref{lem:binom}, the space is at most
\begin{align*}
& \sum_{i = 1}^{100\phi} J_i(\log n - \log J_i + O(1)) + \sum_{i = 1}^{n / \phi} A_i \left(\log (\phi 2^s) - \log A_i + \log e + o(1)\right) + J_0 \left(\log (n2^s) - \log J_0 + O(1)\right) + o(n)\\
& = \sum_{i = 1}^{100 \phi} J_i(\log n - \log J_i) + \sum_{i = 1}^{n / \phi} A_i \left(\log (\phi 2^s) - \log A_i + \log e\right) + J_0 \left(\log (n2^s) - \log J_0\right) + o(n)\\
& = \sum_{i = 1}^{100 \phi} J_i(\log n - \log J_i) + \sum_{i = 1}^{n / \phi} A_i \left(\log (\phi 2^s) - \log A_i\right) + J_0 \left(\log (n2^s) - \log J_0\right) + n \log e +  o(n).
\end{align*}
By Jensen's inequality, the above quantity is maximized by setting $J_1, \ldots, J_\phi$s to be equal (so $J_i \ge \frac{n}{2\phi^2}$ for all $i$) and by setting  $A_1, \ldots, A_{n /\phi}$ to be equal (so $A_i = (n - J_0) / (n / \phi)$ for all $i$). Thus the number of bits used by the hash table is at most
\begin{align*}
& \left(\sum_i J_i\right) \cdot \left(\log n - \log \frac{n}{2\phi^2}\right) + \left(\sum_i A_i\right) \cdot \left(\log (\phi 2^s) - \log \frac{n - J_0}{n / \phi}\right) + J_0 \left(\log (n2^s) - \log J_0\right) + n \log e +  o(n) \\
& = \frac{n}{\phi} \cdot \left(\log n - \log \frac{n}{2\phi^2}\right) + (n - J_0) \cdot \left(\log (\phi 2^s) - \log \frac{n - J_0}{n / \phi}\right) + J_0 \left(\log (n2^s) - \log J_0\right) + n \log e +  o(n) \\
& = \frac{n}{\phi} \cdot \Theta(\log \phi)+ (n - J_0) \cdot \left(\log (\phi 2^s) - \log \frac{n - J_0}{n / \phi}\right) + J_0 \left(\log (n2^s) - \log J_0\right) + n \log e +  o(n) \\
& = (n - J_0) \cdot \left(\log (\phi 2^s) - \log \frac{n - J_0}{n / \phi}\right) + J_0 \left(\log (n2^s) - \log J_0\right) + n \log e +  o(n) \\
& =  (n - J_0) \left(\log (\phi 2^s) - \log \phi\right) + J_0 \left(\log (n2^s) - \log J_0\right) + n \log e +  o(n) \\
& =  (n - J_0)s + J_0 \left(\log (n2^s) - \log J_0\right) + n \log e +  o(n) \\
& =  ns + J_0 \left(\log n - \log J_0\right) + n \log e +  o(n).
\end{align*}
Since $J_0 = o(n)$, this is
$$ns + n \log e +  o(n),$$
which by Lemma \ref{lem:binom} is 
$$\log \binom{n2^s}{n} + o(n).$$
\end{proof}

Since insertions/deletions take time $O(k)$ and queries take time $O(1)$, the preceding lemma implies Theorem \ref{thm:small}. 

We conclude the section with several simple corollaries. The first corollary extends the theorem to store key-value pairs.

\begin{corollary}
Let $k, n$ be parameters. Let $\phi = \Theta(\log^{(k)} n)$. Let $ U = [2^{\log n + s_1}]$ and $V = [2^{s_2}]$ for some $s_1, s_2$ satisfying $s_1 \ge \omega(1)$ and $s_1 + s_2 \le o(\log n) \cap O(\log n / \phi)$.

There exists a fixed-capacity dictionary that stores up to $n$ keys from $U$ at a time, each of which is associated with a value in $V$; that supports insertions/deletions in time $O(k)$ (with high probability in $n$); that supports queries in time $O(1)$; and that uses a total of 
$$\log \binom{|U|}{n} + n s_2 + o(n)$$
bits of space (with high probability in $n$).
\label{cor:smalla}
\end{corollary}
\begin{proof}
This follows from the same sequence of arguments as before, but now we associate values with keys as well.
\end{proof}

The second corollary extends the theorem to support dynamic resizing. Note that, since we are interested in keys whose lengths are very close to $n$, it does not make sense to talk about $n$ changing by a large factor (indeed, this would take us out of the small-key regime). Thus, we focus on $n$ in a range $[N / 2, N]$ for some $N$. 

\begin{corollary}
Let $k, n$ be parameters. Let $\phi = \Theta(\log^{(k)} n)$. Let $ U = [2^{\log n + s_1}]$ and $V = [2^{s_2}]$ for some $s_1, s_2$ satisfying $s_1 \ge \omega(1)$ and $s_1 + s_2 \le o(\log n) \cap O(\log n / \phi)$.

There exists a fixed-capacity dictionary that stores up to $N$ keys from $U$ at a time, each of which is associated with a value in $V$; that supports insertions/deletions in time $O(k)$ (with high probability in $n$); that supports queries in time $O(1)$; and that (with high probability in $n$) uses a total of 
$$\log \binom{|U|}{n} + n s_2 + o(n)$$
bits of space if $n \in [N / 2, N]$ is the number of keys currently present.
\label{cor:smalla2}
\end{corollary}
\begin{proof}
This guarantee is already true of our current data structure. Indeed, every component of the data structure except for the array $A$ of counters is a dynamically-resizable hash table. The array $A$ of counters takes $o(n)$ space, so we only need to worry about the total space consumed by the dynamically-resizable hash tables. The proof of Theorem \ref{thm:small} immediately extends to arbitrary $n \in [N / 2, N]$ to bound the total space by
$$\log \binom{u2^{s_1}}{n} + n s_2 + o(n)$$
bits. 
\end{proof}

\subsection{Constructing optimal filters}\label{sec:filters}

In this section, we apply our results to the problem of constructing space-efficient filters. A filter has three parameters: a maximum capacity $n $, and a false positive rate $\epsilon$ (which we will assume is an inverse power of two), and a universe $U$ of keys. A filter must support insertions/deletions/queries on a dynamic set $S \subseteq U$ of up to $n$ keys. Unlike a dictionary, however, a filter is permitted to sometimes return false positives on queries: if a key $x \not\in S$ is queried, the filter must correctly return that $x \not\in S$ with probability $1 - \epsilon$, but it is permitted to incorrectly return that $x \in S$ with probability $\epsilon$. 

Information theoretically, a static filter (i.e., a filter that supports only queries) must use at least $n  \log \epsilon^{-1} $ bits. It is known \cite{pagh2013approximate} that there exist values of $\epsilon$ for which a dynamic filter must use at least
$n  \log \epsilon^{-1} + \Omega(n)$ bits, but it remains an open question whether there exists a dynamic filter that uses at most $n \log \epsilon^{-1}  + O(n)$ bits for all $\epsilon$. We now establish that, as long as $\log \epsilon^{-1}$ is slightly sublogarithmic in $n$, then such a dynamic filter does, in fact, exist. We also give extremely succinct filters for the setting where $\epsilon^{-1} = \Theta(\log n)$, bringing the number of wasted bits per key to be the same as what we have achieved for the dictionary problem.

\paragraph{Reducing the filter problem to the dictionary problem.}
We begin by reviewing the standard technique for constructing a filter using a dictionary (see, e.g., \cite{carter1978exact, pagh2005optimal, liu2020succinct, arbitman2010backyard, bercea2020dynamic,bender2018bloom}). We hash keys $x\in U $ to $(\log n + \log \epsilon^{-1})$-bit fingerprints $f(x)$. We store the fingerprints $\{f(x) \mid x \in S\}$ in a hash table, and to answer a query for a key $x$, we simply check whether $f(x)$ is in the hash table. If a key $x \not\in S$ is queried, then the probability of a false positive is at most
$$\sum_{y \in S} \Pr[f(x) = f(y)] = n \cdot \frac{\epsilon}{n} = \epsilon.$$

Notice, however, that the fingerprints $\{f(x) \mid x \in S\}$ form a \emph{multi-set}, rather than a set, so we cannot actually store them directly in a hash table. Our solution is to store one copy of each fingerprint in a hash table $\mathcal{A}$, and then to store any duplicate fingerprints in a secondary hash table $\mathcal{B}$ that is capable of supporting multi-sets. Whenever we insert a new key $x$, we first try to place $f(x)$ in $\mathcal{A}$, and if it is already there, we place it in $\mathcal{B}$; whenever we delete a key $x$, we first try to delete (one copy of) $f(x)$ from $\mathcal{B}$, and if it is not there, we delete it from $\mathcal{A}$; and whenever we query a key $x$, we can just check whether $f(x) \in \mathcal{A}$. 

The hash table $\mathcal{B}$ will be significantly smaller than $\mathcal{A}$, meaning that it does not have to be highly space efficient. Thus we are able to use of past work on multi-set dictionaries to implement $\mathcal{B}$:

\begin{lemma}[Theorem 1 of \cite{bercea2020dynamic}]
Let $\epsilon^{-1} \in [\omega(1), O(\log n)]$. There exists a high-probability constant-time hash table that stores an arbitrary multi-set of $m$ keys in $(1 + o(1)) m \log \epsilon^{-1}$ bits.
\label{lem:backyardfilter}
\end{lemma}

In fact, Lemma \ref{lem:backyardfilter} is stronger than what we need---it would suffice for us to have a multi-set dictionary using $o(m \epsilon^{-1})$ bits. Indeed, we can bound $m = |\mathcal{B}|$ by $O(\epsilon n + \log n)$ with high probability:

\begin{lemma}
At any given moment $D = |\{x \in S \mid f(x) = f(y) \text{ for some }y \in S\setminus\{x\}\}|$ satisfies $D = O(\epsilon n + \log n)$ with high probability in $n$.
\label{lem:duplicatekeys}
\end{lemma}
\begin{proof}
Let $x_1, \ldots, x_n$ denote the keys in $S$, and let $Y_i$ be the 0-1 random variable indicating whether $f(x_i) = f(x_j)$ for some $j < i$. Notice that $D \le 2 \sum_i Y_i$. 

The $Y_i$s are independent, and each $Y_i$ satisfies $\Pr[Y_i = 1] \le \epsilon$. Therefore we can apply a Chernoff bound to deduce that $D = O(\epsilon n + \log n)$ with high probability in $n$.
\end{proof}

By Lemma \ref{lem:duplicatekeys}, if $\epsilon = o(1)$, then the total number of bits used by $\mathcal{B}$ is $o(n)$ with high probability in $n$. On the other hand, if $\mathcal{A}$ is implemented using a hash table that wastes $r$ bits per key, then it uses a total of at most
$$\log \binom{n\epsilon^{-1}}{n} + nr$$
bits. If we again assume that $\epsilon = o(1)$, then by Lemma \ref{lem:binom}, this is equal to 
$$ n \log \epsilon^{-1}  + nr + n \log e$$
bits.

Applying Theorems \ref{thm:largekeys} and \ref{thm:small} to construct $\mathcal{A}$, we arrive at the following result:
\begin{theorem}
Let $\epsilon^{-1} \in [\omega(1), O(\log n)]$ be an inverse power of two, and let $k \in [\log^* n]$ be a parameter. One can construct a filter that has false-positive rate at most $\epsilon$, that supports queries in constant time, that supports insertions/deletions in time $O(k)$, and that uses space at most
$$\begin{cases}
n \log \epsilon^{-1}  + n\log e + o(n) & \text{ if  } \epsilon^{-1} \le \frac{\log n}{\log^{(k)} n} \text{ and } \log^{(k)} n = \omega(1) \\
n \log \epsilon^{-1}  + O(n) + O(n \log^{(k)} n) & \text{ otherwise}
\end{cases}$$
bits. The time and space guarantees hold for each operation with high probability in $n$.
\label{thm:filter}
\end{theorem}

We remark that the size $|U|$ of the universe does not matter, since we can use hash function with $O(\log n)$ description bits (see, e.g., Theorem 14 of \cite{pagh2009dispersing}) to map $|U|$ to a universe of size $\poly(n)$, while avoiding collisions with probability $1 - 1 / \poly(n)$. The $1 / \poly(n)$ collision probability can then easily be absorbed into $\epsilon$.\footnote{Notice that the same approach is not legal for hash tables, since the failure probability for a hash table must be with respect to \emph{running time}, rather than with respect to \emph{correctness}. That is, hash tables are never allowed to return false positives.}

\section*{Acknowledgments}

This research was supported in part by NSF grants CSR-1938180, CCF-2106999, CCF-2118620, CCF-2118832, CCF-2106827, CCF-1725543, CSR-1763680, CCF-1716252 and CNS-1938709, as well as an NSF GRFP fellowship, a Fannie and John Hertz Fellowship, and Singapore Ministry of Education (AcRF) Tier 2 grant MOE2018-T2-1-013.

This research was also partially sponsored by the United States Air Force Research Laboratory and the United States Air Force Artificial Intelligence Accelerator and was accomplished under Cooperative Agreement Number FA8750-19-2-1000. The views and conclusions contained in this document are those of the authors and should not be interpreted as representing the official policies, either expressed or implied, of the United States Air Force or the U.S. Government. The U.S. Government is authorized to reproduce and distribute reprints for Government purposes notwithstanding any copyright notation herein.

\end{document}